%% file: main.tex
\documentclass{article}
\usepackage[T1]{fontenc}

\usepackage{ijcai25}

\usepackage{times}
\usepackage{soul}
\usepackage{url}
\usepackage[hidelinks]{hyperref}
\usepackage[utf8]{inputenc}
\usepackage[small]{caption}
\usepackage{graphicx}
\usepackage{amsmath}
\usepackage{comment}
\usepackage{amsthm}
\usepackage{booktabs}
\usepackage{algorithm}
\usepackage{algorithmic}
\usepackage[switch]{lineno}
\usepackage{todonotes}

\newtheorem{definition}{Definition}[section]
\newtheorem{example}[definition]{Example}
\newtheorem{theorem}[definition]{Theorem}
\newtheorem{corollary}[definition]{Corollary}
\newtheorem{lemma}[definition]{Lemma}

\newtheorem{remark}[definition]{Remark}

\title{
Model Checking Linear Temporal Logic with Standpoint Modalities}

\author{
Rajab Aghamov$^1$ \and Christel Baier$^1$ \and Toghrul Karimov$^2$ \and Rupak Majumdar$^2$ \and Joël Ouaknine$^2$ \and Jakob Piribauer$^{1,3}$ \and Timm Spork$^1$ \\
\affiliations
$^1$Technische Universität Dresden, Dresden, Germany\\
$^2$Max Planck Institute for Software Systems, Saarbrücken, Germany \\
$^3$Universität Leipzig, Germany\\
\emails
\{rajab.aghamov, christel.baier\}@tu-dresden.de, \{toghs, rupak, joel\}@mpi-sws.org, jakob.piribauer@uni-leipzig.de, timm.spork@tu-dresden.de
}

\input{macros.tex}

\usepackage{apxproof}
\usepackage{graphicx}
\usepackage{subfigure}
\usepackage{makeidx}
\usepackage{pgf}
\usepackage{color}

\usepackage{amsmath}
\usepackage{amsfonts}
\usepackage{amssymb}
\usepackage{mathptmx}
\usepackage{tikz}
\usetikzlibrary{arrows, automata}

\begin{document}

\maketitle

\input{Sections/abstract.tex}

\input{Sections/introduction.tex}

\input{Sections/preliminaries.tex}

\input{Sections/standpoint-LTL.tex}

\input{Sections/model-checking.tex}

\input{Sections/complexity.tex}

\input{Sections/conclusion.tex}


%\bibliographystyle{named}
\section*{Acknowledgments}
We thank our colleagues for their valuable discussions. Joël Ouaknine is also part of Keble College, Oxford, as an Emmy.network Fellow. The authors were supported by DFG grant 389792660 as part of TRR 248 (see \url{https://perspicuous-computing.science}) and by BMBF (Federal Ministry of Education and Research) in DAAD project 57616814 (SECAI, School of Embedded and Composite AI) as part of the program Konrad Zuse Schools of Excellence in Artificial Intelligence.

\newpage

\bibliographystyle{plain}
\bibliography{references}

%%%%%%%%%%%%%%%%%%%%%%%%%%%%%%%%%%%%%%%%

\newpage

\input{Sections/appendix.tex}

\end{document}

%% file: macros.tex
\newcommand{\optional}[1]{#1}

\newcommand{\Tree}{\mathit{Tree}}
\newcommand{\StpNodes}{\mathit{StpNod}}

\newcommand{\pow}{\mathit{pow}}

\newcommand{\obs}{\mathit{obs}}

\newcommand{\Obsset}{\mathfrak{O}}

\newcommand{\ad}{\mathit{ad}}

\newcommand{\StMod}[2]{\langle \!\! \langle #1 \rangle \!\!\rangle #2}
\newcommand{\DualStMod}[2]{[\![ #1 ]\!] #2}

\newcommand{\Ag}{\mathit{Ag}}

\newcommand{\proj}[2]{#1|_{#2}}
\newcommand{\Nat}{\mathbb{N}}

\newcommand{\StpLTL}{\mathrm{SLTL}}

\newcommand{\true}{\mathit{true}}

\newcommand{\neXt}{\bigcirc}
\newcommand{\Until}{ \, \mathrm{U}\, }

\newcommand{\primed}[1]{#1'}

\newcommand{\ltlk}[1]{#1_{\text{\rm \tiny LTLK}}}
\newcommand{\ltlkstep}[1]{#1_{\text{\rm \tiny LTLK}}^{\textrm{\tiny step}}}

\newcommand{\LTLmodels}{\models_{\text{\rm \tiny LTL}}}
\newcommand{\LTLKmodels}{\models_{\text{\rm \tiny LTLK}}}

\newcommand{\Stmodels}[1]{\models_{\scriptscriptstyle #1}}
\newcommand{\QStmodels}[2]{\models_{\scriptscriptstyle #1}^{\scriptscriptstyle #2}}

\newcommand{\step}{\text{\rm step}}
\newcommand{\stepmodels}{\Stmodels{\step}}

\newcommand{\probs}{\text{\rm pobs}}
\newcommand{\pobs}{\probs}
\newcommand{\probsmodels}{\Stmodels{\pobs}}
\newcommand{\pobsmodels}{\probsmodels}

\newcommand{\public}{\text{\rm public}}
\newcommand{\publicmodels}{\Stmodels{\public}}

\newcommand{\decr}{\text{\rm decr}}
\newcommand{\decrmodels}[1]{\QStmodels{\decr}{\scriptscriptstyle #1}}

\newcommand{\incr}{\text{\rm incr}}
\newcommand{\incrmodels}[1]{\QStmodels{\incr}{\scriptscriptstyle #1}}

\newcommand{\tinydecr}{\scriptscriptstyle {\tiny \decr}}
\newcommand{\tinyincr}{\scriptscriptstyle {\tiny \incr}}
\newcommand{\tinypublic}{\scriptscriptstyle {\tiny \public}}
\newcommand{\tinypobs}{\scriptscriptstyle {\tiny \pobs}}
\newcommand{\tinystep}{\scriptscriptstyle {\tiny \step}}

\newcommand{\InnerModels}[1]{\models_{\scriptscriptstyle #1}}

\newcommand{\SLTL}[2]{\StpLTL^{\text{\tiny #1}}_{#2}}

\newcommand{\Sat}{\mathit{Sat}}

\newcommand{\Hist}{\mathit{Hist}}

\newcommand{\PSPACE}{\mathit{PSPACE}}

\newcommand{\Paths}{\mathit{Paths}}
\newcommand{\Traces}{\mathit{Traces}}
\newcommand{\trace}{\mathit{trace}}

\newcommand{\Reach}{\mathit{Reach}}

\newcommand{\init}{\mathit{init}}
\newcommand{\Init}{\mathit{Init}}

\newcommand{\last}{\mathit{last}}
\newcommand{\first}{\mathit{first}}

\newcommand{\suffix}[2]{#1[#2\ldots \infty]}
\newcommand{\prefix}[2]{#1[0\ldots #2]}

\newcommand{\cA}{\mathcal{A}}

\newcommand{\cD}{\mathcal{D}}

\newcommand{\cL}{\mathcal{L}}

\newcommand{\cP}{\mathcal{P}}

\newcommand{\cT}{\mathcal{T}}

\newcommand{\cX}{\mathcal{X}}

\newcommand{\fT}{\mathfrak{T}}

%% file: Sections/abstract.tex
\begin{abstract}
	Standpoint linear temporal logic ($\StpLTL$) is a recently introduced extension of classical linear temporal logic (LTL) with standpoint modalities. Intuitively, these modalities allow to express that, from agent $a$'s standpoint, it is conceivable that a given formula holds. 
	Besides the standard interpretation of the standpoint modalities  we introduce four  new semantics, which differ in the information an agent can extract from the  history. We provide a general model checking algorithm applicable to $\StpLTL$ under any of the five semantics. Furthermore we analyze the computational complexity of the corresponding model checking problems, obtaining PSPACE-completeness in three cases, which stands in contrast to the known EXPSPACE-completeness of the $\StpLTL$ satisfiability problem.
\end{abstract}

%% file: Sections/introduction.tex
%%%%%%%%%%%%%%%%%%%%%%%%%%%%%%%%%%%%%%%%%%%%%%%%%%%%%%%%%%%%%%%%%%%
%%%%%%%%%%%%%%%%%%%%%%%%%%%%%%%%%%%%%%%%%%%%%%%%%%%%%%%%%%%%%%%%%%%
%%%%
%%%%     introduction
%%%%
%%%%%%%%%%%%%%%%%%%%%%%%%%%%%%%%%%%%%%%%%%%%%%%%%%%%%%%%%%%%%%%%%%%
%%%%%%%%%%%%%%%%%%%%%%%%%%%%%%%%%%%%%%%%%%%%%%%%%%%%%%%%%%%%%%%%%%%

\section{Introduction}

Automated reasoning about the dynamics of scenarios in which multiple agents with access to different information interact is a key problem
in artificial intelligence and  formal verification. Epistemic temporal logics  are prominent, expressive formalisms to specify properties of such scenarios (see, e.g., \cite{HalVar-STOC86,Halpern-Overview86,Reasoning-about-knowledge-MIT-2004,MC-LTLK-and-beyond-2024}).
The resulting algorithmic problems, however, often have  non-elementary  complexity or are even undecidable (see, e.g., \cite{vanderMeydenS1999,Dima2009,MC-LTLK-and-beyond-2024}).

Aiming to balance expressiveness and computational tractability, \cite{SL-FOIS21,phdGomez} defines
static \emph{standpoint logics} that extend propositional logic with modalities
$\StMod{a}{\varphi}$ expressing that ``according to agent $a$, it is conceivable that $\varphi$'' and the dual modalities
$\DualStMod{a}{\varphi}$ expressing that ``according to $a$, it is unequivocal that $\varphi$''.
Standpoint logics and their extensions have proven useful to, e.g., reason about inconsistent  formalizations of concepts 
in the medical domain to align different ontologies and in a forestry application, where different sources disagree about the global extent of forests \mbox{\cite{SL-FOIS21,phdGomez,Gomez2022}}.

Recently introduced combinations of linear temporal logic (LTL) with standpoint modalities
\cite{SLTL-KR23,Complexity-SLTL-ECAI24} enables reasoning about dynamical aspects of multi-agent systems.
The focus of \cite{SLTL-KR23,Complexity-SLTL-ECAI24}  
is the satisfiability problem for the resulting  \emph{standpoint LTL} ($\StpLTL$).
In this paper, we consider the model-checking problem that asks whether all executions of a transition system satisfy a given $\StpLTL$-formula. To the best of our knowledge, this problem has not been addressed in the literature.

Whether the formula $\StMod{a}{\varphi}$ holds  after some finite history, i.e., whether it is plausible for agent $a$ that property $\varphi$ holds on the future execution, depends on $a$'s standpoint
 as well as what was observable to $a$ from the history. 
 To illustrate this, consider a situation in which different political agents have different perceptions of how actions taken by the state influence future developments.
 These perceptions are the standpoints of the agents. After a series of events, i.e., a history, the agents deem different future developments possible according to their standpoint
 as well as the aspects of the history that they are actually aware of.  When reasoning about each other's standpoints, there are  different ways in which information might be exchanged between the agents.
 For example, in a discussion agents might learn about aspects of the history they were not aware of; or they might completely ignore what other agents have observed in the past.
 
 Besides the semantics for $\StpLTL$ proposed in \cite{SLTL-KR23,Complexity-SLTL-ECAI24}, we introduce four additional semantics that differ in the amount of information agent $a$ can access from the history and how information is transferred between agents.
  To formalize these semantics, we follow  a natural
approach   for standpoint logics that
uses separate
transition systems $\cT_a$ describing the executions that are consistent with $a$'s standpoint  as in 
\mbox{\cite{SLTL-KR23,Complexity-SLTL-ECAI24}}, together with a main transition system $\cT$ modeling the actual system.
Unlike \cite{SLTL-KR23,Complexity-SLTL-ECAI24}, 
we assume the labels of the states in $\cT_a$ to be from 
a subset $P_a$ of the set $P$ of atomic propositions of $\cT$.
The difference
to the  use of indistinguishablity
relations $\sim_a$ on states of a single transition system for all agents $a \in \Ag$, common in epistemic temporal logics,
 is mostly
of syntactic nature as we will show by translations of families of 
indistinguishablity
relations $(\sim_a)_{a \in \Ag}$ on a  transition system to 
separate transition systems $(\cT_a)_{a \in \Ag}$,  and vice versa.

Under all five semantics, the intuitive meaning of $\StMod{a}{\varphi}$ is that
there is a state $s$ in $\cT_a$, which is one of the 
potential current states from agent $a$'s view of the history, 
and a path $\pi$ in $\cT_a$ from $s$ such that $\pi$ satisfies $\varphi$
when agent $a$ makes nondeterministic guesses for truth values of 
the atomic propositions outside $P_a$. Informally, the different semantics are as follows:

\noindent
  -- The \emph{step semantics} $\stepmodels$   agrees with the 
  semantics  proposed 
  in \cite{SLTL-KR23,Complexity-SLTL-ECAI24}. 
  It assumes that only
  the number of steps  performed in the past are accessible
  to the agents.

\noindent
  -- The \emph{pure observation-based semantics} $\probsmodels$ 
  is in the spirit of
  the perfect-recall LTLK (LTL extended with knowledge operators) semantics  of \cite{MC-LTLK-and-beyond-2024} 
  where  $a$ can access 
  exactly the truth values of
  the atomic propositions in $P_a$ from the history.

\noindent
  -- The \emph{public-history semantics} $\publicmodels$ can be seen
  as a perfect-recall variant of the LTLK semantics 
  where all agents have
  full access to the history.

\noindent
  -- The \emph{decremental semantics} $\decrmodels{}$ is a variant of
  $\probsmodels$ where 
  standpoint subformulas $\StMod{b}{\psi}$ 
  of a  formula $\StMod{a}{\varphi}$ 
  are interpreted from the view of agent $a$, when $a$ knows the transitions
  system of $b$, but can access only 
  the atomic propositions
  in $P_a \cap P_b$ to guess what agent $b$ knows from the history.

\noindent
  -- The \emph{incremental semantics} $\incrmodels{}$ is as $\decrmodels{}$, but under the assumption that standpoint subformulas $\StMod{b}{\psi}$ 
  of a standpoint formula $\StMod{a}{\varphi}$ 
  are interpreted from the view of the coalition $\{a,b\}$, i.e.,
  that $a$ can access
  the atomic propositions
  in $P_a \cup P_b$ to guess what agent $b$ knows from the history.

The decremental and incremental semantics share ideas of distributed knowledge and the ``everybody knows'' operator of epistemic logics \cite{Reasoning-about-knowledge-MIT-2004}.

\paragraph*{Main contributions.} Besides introducing the four new semantics for $\StpLTL$ (Section \ref{sec:logic}), our main contributions are 

\noindent
--  a generic model-checking algorithm that is applicable for all five semantics
  (Section \ref{sec:model-checking})

\noindent
-- complexity-theoretic results for the model checking problem of $\StpLTL$ under the different semantics (Section \ref{sec:complexity}).
More precisely we show
PSPACE-completeness for full $\StpLTL$ under $\stepmodels$ and $\publicmodels$, and for $\StpLTL$ formulas of alternation depth 1 under $\pobsmodels$, $\decrmodels{}$ and $\incrmodels{}$. This stands in contrast to the EXPSPACE-completeness of the satisfiability problem for SLTL under the step semantics \cite{Complexity-SLTL-ECAI24}.
Furthermore, our results yield an EXPTIME upper bound for $\decrmodels{}$. The same holds for $\pobsmodels$ under the additional assumption that the $P_a$'s are pairwise disjoint.
We show that $\StpLTL$ under all five semantics can be embedded into LTLK. For the case of $\incrmodels{}$, the embedding yields an $(N{-}1)$-EXPSPACE upper bound where $N=|\Ag|$. For the case of $\pobsmodels$ and the $\StpLTL$ fragment of alternation depth at most $d$, the embedding into LTLK implies 
$(d{-}1)$-EXPSPACE membership.

While our algorithm relies on similar ideas as the
LTLK model-checking algorithm as in \cite{MC-LTLK-and-beyond-2024} 
(even for the richer logic CTL*K), it
 exploits the simpler nature of
$\StpLTL$ compared to LTLK and generates smaller
history-automata than those that would have been constructed when
applying iteratively 
the powerset constructions of \cite{MC-LTLK-and-beyond-2024}.
As such, our algorithm can be seen as an adaption of 
\cite{MC-LTLK-and-beyond-2024} that takes a more fine-grained approach for the different $\StpLTL$ semantics resulting in the 
different complexity bounds described above.

%%%%%%%%%%%%%%%%%%%%%%%%%%%%%%%%%%%%%%%%%%%%%%%%%%%%

Omitted proofs and details can be found in the appendix.

%% file: Sections/preliminaries.tex
%%%%%%%%%%%%%%%%%%%%%%%%%%%%%%%%%%%%%%%%%%%%%%%%%%%%%%%%%%%%%%%%%%%
%%%%%%%%%%%%%%%%%%%%%%%%%%%%%%%%%%%%%%%%%%%%%%%%%%%%%%%%%%%%%%%%%%%
%%%%
%%%%     Preliminaries: LTL, TS, etc
%%%%
%%%%%%%%%%%%%%%%%%%%%%%%%%%%%%%%%%%%%%%%%%%%%%%%%%%%%%%%%%%%%%%%%%%
%%%%%%%%%%%%%%%%%%%%%%%%%%%%%%%%%%%%%%%%%%%%%%%%%%%%%%%%%%%%%%%%%%%

\section{Preliminaries}

\label{sec:prelim}

Throughout the paper, we assume some familiarity with
linear temporal logic interpreted over transition systems
and automata-based model checking, see e.g.~\cite{CGKPV18,BK08}.

\paragraph*{Notations for strings.}
Given an alphabet $\Sigma$, 
we write $\Sigma^*$ for the set of finite strings over $\Sigma$,
$\Sigma^{\omega}$ for the set of infinite strings over $\Sigma$
and $\Sigma^{\infty}$ for $\Sigma^* \cup \Sigma^{\omega}$.
As usual, $\Sigma^+ = \Sigma^* \setminus \{\varepsilon\}$ where
$\varepsilon$ denotes the empty string.
Given a (in)finite string 
$\varsigma = H_0 \, H_1 \ldots H_n$ or
$\varsigma = H_0 \, H_1 \, \ldots$ over $\Sigma$,
let $\first(\varsigma)=H_0$. 
If $\varsigma = H_0 \ldots H_n$ is finite then 
$\last(\varsigma)=H_n$.
Given $i, j \in \Nat$ we write $\varsigma[i \ldots j]$ 
for the substring $H_i \ldots H_j$ if $i \leqslant j$
(and assuming $j \leqslant n$ if $\varsigma$ is a finite string of length $n$)
and $\varsigma[i \ldots j]=\varepsilon$ if $i > j$.
If $i=j$ then $\varsigma[i \ldots i] = \varsigma[i] = H_i$.
So, $\prefix{\varsigma}{j}$ denotes the prefix $H_0 \ldots H_j$. If
 $\varsigma$ is infinite then $\suffix{\varsigma}{j} = H_j \, H_{j+1} \, H_{j+2} \ldots$.

If $\Sigma = 2^P$ is the powerset of $P$ and $R \subseteq P$ then the projection function 
$\proj{}{R} : (2^P)^{\infty} \to (2^R)^{\infty}$ 
is obtained by applying the projection 
$2^P \to 2^R$, $H \mapsto H \cap R$, 
elementwise, i.e., if $\varsigma = H_0 \, H_1 \, H_2 \ldots$ then
$\proj{\varsigma}{R} = 
   (H_0 \cap R) \, (H_1 \cap R) \, (H_2 \cap R) \ldots$.

%%%%%%%%%%%%%%%%%%%

\paragraph*{Transition systems.}
A transition system is a tuple $\cT = (S,\to,\Init,R,L)$ where
$S$ is a finite state space, ${\to} \, \subseteq \, S \times S$ a total transition
relation (where totality means that every state $s$ has at least one outgoing transition $s \to s'$), $\Init \subseteq S$ the set of initial states, $R$ a finite set of atomic propositions and $L : S \to 2^R$ the labeling function.
If $\Init$ is a singeleton, say $\Init = \{\init\}$, we simply write
$\cT = (S,\to,\init,R,L)$.

A path in $\cT$ is a (in)finite string $\pi = s_0 \, s_1 \ldots s_n \in S^+$ or $\pi = s_0 \, s_1 \, s_2 \ldots \in S^{\omega}$ such that $s_i \to s_{i+1}$ for all $i$. $\pi$ is initial if $\first(\pi)\in \Init$. 
The trace of $\pi$ is
$\trace(\pi) = L(s_0) \, L(s_1) \ldots \in (2^R)^+ \cup (2^R)^{\omega}$.
If $s \in S$ then 
$\Paths(\cT,s)$ denotes the set of infinite paths in $\cT$ starting in $s$ and
$\Traces(\cT,s)= \{\trace(\pi) : \pi \in \Paths(\cT,s) \}$.
If $P$ is a superset of $R$ then
\begin{center}
  $\Traces^P(\cT,s) =
     \bigl\{  \rho \in  \bigl(2^{P}\bigr)^{\omega} : 
         \proj{\rho}{R} \in \Traces(\cT,s)  \bigr\}$.
\end{center}
Thus, $\Traces(\cT,s)\subseteq (2^R)^{\omega}$, while
$\Traces^P(\cT,s)\subseteq (2^P)^{\omega}$. 
Moreover,
$\Paths(\cT) = \bigcup_{s\in \Init} \Paths(\cT,s)$.
$\Traces(\cT)$ and $\Traces^P(\cT)$ have the analogous meaning.
If $h \in (2^P)^+$ then $\Reach(\cT,h)$ denotes the
set of states $s$ in $\cT$ that are reachable from $\Init$ via a path $\pi$ 
with $\trace(\pi)=\proj{h}{R}$.

%%%%%%%%%%%

\paragraph*{Linear temporal logic (LTL).}

The syntax of LTL over $P$ is given by (where $p \in P$): 
\begin{center}
 $\begin{array}{lclcr}
  \varphi & ::= & 
  \true \ \ \big| \ \ p \ \ \big| \ \ \neg \varphi \ \ \big| \ \ \varphi_1 \wedge \varphi_2
              \ \ \big| \ \neXt \varphi \ \ \big| \ \ \varphi_1 \Until \varphi_2 
 \end{array}$
\end{center}
Other Boolean connectives are derived as usual, e.g.,
$\varphi_1 \vee \varphi_2 = \neg (\neg \varphi_1 \wedge \neg \varphi_2)$.
The modalities $\Diamond$ (eventually) and $\Box$ (always) are defined by
$\Diamond \varphi = \true \Until \varphi$ and
$\Box \varphi = \neg \Diamond \neg \varphi$.
The standard LTL semantics is formalized by a satisfaction relation
$\LTLmodels$ where formulas are interpreted over infinite traces 
(i.e., elements of $(2^P)^{\omega}$), see e.g. \cite{CGKPV18,BK08}.
An equivalent semantics can be provided
using a satisfaction relation $\models$
that interprets
formulas over trace-position pairs 
$(\rho,n) \in (2^P)^{\omega} \times \Nat$ such that
$(\rho,n) \models \varphi$ iff $\suffix{\rho}{n} \LTLmodels \varphi$.

We use here another equivalent formalisation of the semantics of LTL 
(and later its extension $\StpLTL$) that
interpretes formulas over \emph{future-history pairs} 
$(f,h) \in (2^P)^{\omega}\times (2^P)^+$ with $\last(h)=\first(f)$, see the upper part of Figure \ref{fig:semantics}
where $f[1\ldots 0]=\varepsilon$.
Then, $f \LTLmodels \varphi$ iff $(f, \first(f)) \models \varphi$.
\begin{figure*}[ht]
\begin{center}
  \begin{tabular}{l}
   \begin{tabular}{ll}
   \begin{tabular}{lcl}
     $(f,h) \models p$ & iff & $p \in f[0]$ 
     \\[1ex]

     $(f,h) \models \varphi_1 \wedge \varphi_2$ & iff &
     $(f,h) \models \varphi_1$ and $(f,h)\models \varphi_2$ 
     \\[1ex]
    \end{tabular}
    & \hspace*{0.5cm}
    \begin{tabular}{lcl}
     $(f,h) \models \neg \varphi$ & iff &
     $(f,h) \not\models \varphi$
     \\[1ex]

     $(f,h) \models \neXt \varphi$ & iff &
     $(\suffix{f}{1},h f[1]) \models \varphi$
     \\[1ex]
    \end{tabular}
   \end{tabular}
    \\
 
   \begin{tabular}{l}
    \begin{tabular}{lcl}
     $(f,h) \models \varphi_1 \Until \varphi_2$ & iff &
     there exists $\ell \in \Nat$ such that
     $(\suffix{f}{\ell}, h f[1\ldots \ell]) \models \varphi_2$
     and
     $(\suffix{f}{j}, h f[1\ldots j]) \models \varphi_1$
     for all $j < \ell$
     \\[1ex]

     $(f,h) \models \StMod{a}{\varphi}$ & iff &
     there exists $h' \in (2^{P})^+$, 
     $t\in \Reach(\cT_a,h')$ and
     $f'\in \Traces^P(\cT_a,t)$ such that
     $\last(h')=\first(f')$, \\
     & &
     $\obs_a(h)=\obs_a(h')$ and $(f',h') \models \varphi$ 
   \end{tabular}
   \end{tabular}
  \end{tabular}
\end{center}
\vspace{-10pt}
\caption{Satisfaction relation $\models$ for $\StpLTL$
    over future-history pairs $(f,h)\in (2^P)^{\omega}\times (2^P)^+$
    with $\last(h)=\first(f)$.
    \vspace{-10pt}}
\label{fig:semantics}
\end{figure*}
For interpreting LTL formulas, the history is irrelevant: $f \LTLmodels \varphi$ iff $(f, \first(f)) \models \varphi$
iff $(f,h) \models \varphi$ for some $h$ with
$\last(h)=\first(f)$
iff $(f,h) \models \varphi$ for all $h$ with
$\last(h)=\first(f)$.

If $\cT=(S,\to,\Init,R,L)$ is a transition system with $R \subseteq P$ 
and $\varphi$ an LTL formula over $P$ then
$\cT \LTLmodels \varphi $ iff $f \LTLmodels \varphi$ 
for each $f \in \Traces^P(\cT)$.
$\Sat_{\cT}(\exists \varphi)$ denotes the set of states $s\in S$ where
$\{ f \in \Traces^P(\cT,s) : f \LTLmodels \varphi\} \not= \varnothing$.

%% file: Sections/standpoint-LTL.tex
%%%%%%%%%%%%%%%%%%%%%%%%%%%%%%%%%%%%%%%%%%%%%%%%%%%%%%%%%%%%%%%%%%%
%%%%%%%%%%%%%%%%%%%%%%%%%%%%%%%%%%%%%%%%%%%%%%%%%%%%%%%%%%%%%%%%%%%
%%%%
%%%%     Syntax and semantics of StpLTL
%%%%
%%%%%%%%%%%%%%%%%%%%%%%%%%%%%%%%%%%%%%%%%%%%%%%%%%%%%%%%%%%%%%%%%%%
%%%%%%%%%%%%%%%%%%%%%%%%%%%%%%%%%%%%%%%%%%%%%%%%%%%%%%%%%%%%%%%%%%%

\section{$\StpLTL$: LTL with standpoint modalities}

\label{sec:logic}

Standpoint LTL ($\StpLTL$)
extends LTL by standpoint modalities $\StMod{a}{\varphi}$ 
where $a$ is an agent and $\varphi$ a formula.

\subsection{Syntax}

Given a finite set $P$ of atomic propositions and
a finite set $\Ag$ of agents, say $\Ag= \{a, b, \ldots\}$,
the syntax of $\StpLTL$ formulas over $P$ and $\Ag$ for $p \in P$ and $a \in \Ag$ is given by
\begin{center}
 $\begin{array}{lclcr}
  \varphi & ::= & 
  \true  \ \big|  \ p \ \big| \ \neg \varphi  \ \big| \ \varphi_1 \wedge \varphi_2
  \ \big| \ \neXt \varphi  \ \big|  \ \varphi_1 \Until \varphi_2  \ \big|  \ 
   \StMod{a}{\varphi}
 \end{array}$
\end{center}
The intuitive meaning of $\StMod{a}{\varphi}$
is that from agent $a$'s standpoint it is conceivable that $\varphi$ 
will hold, in the sense that there are indications from $a$'s view 
that there is a path
starting in the current state that fulfills $\varphi$. 
The dual standpoint modality is defined by
$\DualStMod{a}{\varphi} \ = \ \neg \StMod{a}{\neg \varphi}$ and
has the intuitive meaning that from the standpoint of $a$,
$\varphi$ is unequivocal.
That is, under $a$'s view
all paths starting in the current state fulfill $\varphi$.

Formulas of the shape $\StMod{a}{\varphi}$ are 
called \emph{standpoint formulas}.
If $\varphi$ is a $\StpLTL$ formula, then \emph{maximal standpoint subformulas} of $\varphi$
are subformulas that have the form $\StMod{a}{\phi}$ 
and that are not in the scope
of another standpoint operator.
For example, 
$\varphi = (p \wedge \neXt \chi_1) \vee \psi$
with $\chi_1 = \StMod{a}{ (\StMod{b}{q} \Until r)}$ and
 $\psi =
     \DualStMod{c}{ \neXt (r \wedge \StMod{a}{p \Until \StMod{a}{q}})}$
has two maximal standpoint subformulas, namely
$\chi_1$ and
$\chi_2 = 
 \StMod{c}{ \neg \neXt (r \wedge \StMod{a}{p \Until \StMod{a}{q}})}$.

The \emph{alternation depth} $\ad(\varphi)$ of $\varphi$
is the maximal number of alternations between standpoint modalities
for different agents. 
For instance, 
$\ad(\StMod{a}{(p \wedge \neXt \DualStMod{b}{q})})=2$, while
$\ad(\StMod{a}{(p \wedge \neXt \DualStMod{a}{q})})= 
 \ad(\StMod{a}{p} \wedge \neXt \DualStMod{b}{q}) = 1$ 
if $a \not= b$.
The precise definition is provided in
Definition \ref{def:alternation-depth} in the appendix.
Each $\StpLTL$ formula over a singleton agent set $\Ag=\{a\}$ has
alternation depth at most 1. 
For $d \in \Nat$, let $\StpLTL_d$ denote the sublogic of $\StpLTL$ where all formulas $\varphi$ satisfy $\ad(\varphi) \leq d$. In particular, $\StpLTL_0$ is LTL.

\begin{remark}
\label{remark:sharpening}
{\rm
The original papers \cite{SLTL-KR23,Complexity-SLTL-ECAI24} 
on standpoint LTL have an additional 
type of formula $a \preceq b$ 
where $a, b \in \Ag$.
The intuitive semantics of $a \preceq b$ is that 
the standpoint of $a$ is sharper than that of $b$.
In our setting, where we are given
transition systems for $a$ and $b$, formulas $a \preceq b$ are
either true or false depending on whether the set of traces
of the transition system representing $a$'s view is contained in
the set of $b$'s traces or not. 
As the trace inclusion problem is PSPACE-complete \cite{KanSmo90}
like the LTL model-checking problem
\cite{SistlaClarke-JACM85},
the complexity results of Section \ref{sec:complexity} 
are not affected when adding sharpening statements $a \preceq b$ to $\StpLTL$.
  }
\end{remark}

%%%%%%%%%%%%%%%%%%%%%%%%%%%%%%%%%%%%%%%%%%%%%%%%%%%%%%%%%%%%%%%%%%%%%%%%%%
%%%%%
%%%%%      subsection: semantics
%%%%%
%%%%%%%%%%%%%%%%%%%%%%%%%%%%%%%%%%%%%%%%%%%%%%%%%%%%%%%%%%%%%%%%%%%%%%%%%%

\subsection{Semantics of $\StpLTL$}
\label{sec:semantics}

We will consider different semantics of the standpoint modality
that differ in what the agents can observe from the history.

%%%%%%%%%%%%%%%%%%%%%%%%%%%%%%%%%%%%%%%%%%%%%%%%%%%%%%%%%%%%%%%%%%%%%%%%%%

\paragraph*{$\StpLTL$ structures.}
\label{sec:structures}
$\StpLTL$ structures are
tuples 
$\fT = \bigl(\cT_0,(\cT_a)_{a\in \Ag}\bigr)$ with $\cT_0 = (S_0,\to_0,\init_0,P_0,L_0)$ a 
transition system 
over the full set $P_0=P$ of atomic propositions, and
$\cT_a= (S_a,\to_a,\init_a,P_a,L_a)$ transition systems 
for the agents $a\in \Ag$ over some $P_a \subseteq P$.
For simplicity, we assume these transition systems to have
unique initial states.

%%%%%%%%%%%%%%%%%%%%%%%%%%%%%%%%%%%%%%%%%%%%%%%%%%%%%%%%%%%%%%%%%%%

\paragraph*{Semantics of the standpoint modalities.}

We consider five different semantics for
$\StMod{a}{\varphi}$, using satisfaction relations
$\stepmodels$ (step semantics as in \cite{SLTL-KR23,Complexity-SLTL-ECAI24}), 
$\probsmodels$ 
(pure observation-based semantics which is essentially the
perfect-recall partial information semantics for LTLK
\cite{MC-LTLK-and-beyond-2024} 
adapted for $\StpLTL$ structures),
$\publicmodels$ 
(public-history semantics which can be seen as a 
perfect-recall semantics where agents have full information about the
history),
$\decrmodels{Q}$ and $\incrmodels{Q}$
(variants of $\pobsmodels$ with decremental resp. incremental knowledge)
where $Q \subseteq P$.

We deal here with an interpretation over future-history pairs 
(see Section \ref{sec:prelim}).
Let $\models$ be one of the five satisfaction relations.
The semantics of $\StpLTL$
can be presented in a uniform manner
as shown in  Figure \ref{fig:semantics}.
For the dual standpoint operator we obtain:
$(f,h) \models \DualStMod{a}{\varphi}$ iff $(f',h') \models \varphi$ 
for all $h' \in (2^{P})^+$, $t\in \Reach(\cT_a,h')$ 
and $f'\in \Traces^P(\cT_a,t)$ with
$\last(h')=\first(f')$ and
$\obs_a(h)=\obs_a(h')$.

The five semantics rely on different observation functions
$\obs_a$. In all cases, 
$\obs_a$ is a projection $\obs_a: (2^P)^+ \to (2^{\Obsset_a})^+$, 
$\obs_a(h) = \proj{h}{\Obsset_a}$
for some $\Obsset_a \subseteq P$.
Intuitively, $\Obsset_a$ formalizes which of the propositions are visible to agent $a$
in the history. (For the decremental and
incremental semantics,
both $\Obsset_a$ and the induced observation function $\obs_a$
do not only depend on $a$, but on the context of the
formula $\StMod{a}{\varphi}$ as will be explained later.)
Before presenting the specific choices of $\Obsset_a$ in the five 
$\StpLTL$ semantics, we make some general observations:

\begin{lemma}
\label{future-ind-lemma}
  If $f_1, f_2\in (2^P)^{\omega}$ and $h \in (2^P)^+$ with 
  $\last(h)=\first(f_1)=\first(f_2)$ then
    $(f_1,h) \models \StMod{a}{\varphi}$  iff 
    $(f_2,h) \models \StMod{a}{\varphi}$. 
\end{lemma}
Lemma \ref{future-ind-lemma} permits to drop the $f$-component 
and to write $(*,h)\models \StMod{a}{\varphi}$
when $(f,h)\models \StMod{a}{\varphi}$ for some (each) future
$f$ with $\first(f)=\last(h)$. 
The truth values of standpoint formulas $\StMod{a}{\varphi}$
only depend on agent $a$'s observation of the history:

\begin{lemma}
 \label{obs-a-lemma}
    If $h_1, h_2 \in (2^P)^+$ with $\obs_a(h_1)=\obs_a(h_2)$  
    then
    $(*,h_1) \models \StMod{a}{\varphi}$ iff 
    $(*,h_2) \models \StMod{a}{\varphi}$.
\end{lemma}

\begin{remark}
\label{remark:standard-SLTL-semantics}
  {\rm
$\StpLTL$ 
structures defined in \cite{SLTL-KR23,Complexity-SLTL-ECAI24}
assign to each agent $a$ a nonempty subset $\lambda(a) \subseteq (2^P)^{\omega}$
and assume a universal agent, called *, such that 
$\lambda(a)\subseteq \lambda(*)$ for all other agents $a$.
The latter is irrelevant for our purposes.
Assuming transition system representations for the sets $\lambda(a)$
is natural for the model checking problem.
In contrast to \cite{SLTL-KR23,Complexity-SLTL-ECAI24}, 
we suppose here that the 
standpoint transition systems $\cT_a$ are defined over some $P_a \subseteq P$, which appears more natural for defining the information that an agent can extract
from the history.
In the semantics of $\StMod{a}{\varphi}$, we thus switch from
$\Traces(\cT_a,t) \subseteq (2^{P_a})^{\omega}$ to
$\Traces^P(\cT_a,t) \subseteq (2^{P})^{\omega}$ which essentially means
that $a$ may guess the truth values of the atomic propositions
in $P \setminus P_a$ to predict whether $\varphi$ can hold in the future.
Alternatively, one could define $\StpLTL$ structures
as tuples $(\cT_0, (\cT_a,P_a)_{a \in \Ag})$ where
$\cT_0$ is as before, the $\cT_a$'s are transition systems 
over some $R_a \subseteq P$, and $P_a \subseteq R_a$ where
the $P_a$ serves to define the functions
$\obs_a$ for the histories.
With $R_a=P$ 
and $\lambda(a)=\Traces(\cT_a)$,
$\StpLTL$ under $\stepmodels$ agrees with the logic 
considered 
in \cite{SLTL-KR23,Complexity-SLTL-ECAI24} (except for sharpening
statements; see Remark \ref{remark:sharpening}).
Our model checking algorithm can easily be
adapted for this more general type of $\StpLTL$ structures 
without affecting our complexity results.
   }
\end{remark}

%%%%%%%%%%%%%%   step semantics

\paragraph*{Step semantics:}
The semantics of the standpoint modality 
$\StMod{a}{\varphi}$ 
introduced in 
\cite{SLTL-KR23,Complexity-SLTL-ECAI24} 
relies on the assumption that the only information that the agents can extract
from the history is the number
of steps that have been performed in the past.
They formulate the semantics in terms of trace-position pairs 
$(\rho,n)\in (2^P)^{\omega}\times \Nat$ and define
$(\rho,n) \models \StMod{a}{\varphi}$ iff $(\rho',n)\models \varphi$
for some $\rho'\in \Traces(\cT_a)$.
Reformulated to our setting,
the set of propositions visible for agent $a$ in the history
is $\Obsset_a=\varnothing$, which yields the 
observation function $\obs_a : (2^P)^+ \to \Nat$, $\obs_a(h)=|h|$.
An equivalent formulation for future-history pairs is:
  $(f,h) \stepmodels \StMod{a}{\varphi}$ iff
     there exist a word $h' \in (2^{P})^+$,
     a state 
     $t\in \Reach(\cT_a,h')$ 
     and a trace
     $f'\in \Traces^P(\cT_a,t)$ such that 
     $\last(h')=\first(f')$,
     $|h|=|h'|$ and $(f',h') \stepmodels \varphi$.

%%%%%%%%    pure observation-based semantics

\paragraph*{Pure observation-based semantics:}
In the style of the (dual of the) classical K-modality in LTLK \cite{HalVar-STOC86,HalVar-JCSS89,MC-LTLK-and-beyond-2024} (cf. Section \ref{sec:embedding-StpLTL-into-LTLK}) we can deal with
the observation function that projects the given history $h$ to the
observations that agent $a$ can make when exactly the propositions in $P_a$
are visible for $a$. That is, for the pure observation-based semantics,  $\Obsset_a=P_a$. Then,
  $(f,h) \probsmodels \StMod{a}{\varphi}$ iff 
     there exist a word $h' \in (2^{P})^+$,
     a state 
     $t\in \Reach(\cT_a,h')$ 
     and a trace
     $f'\in \Traces^P(\cT_a,t)$ such that 
     $\last(h')=\first(f')$,
     $\proj{h}{P_a}=\proj{h'}{P_a}$ and $(f',h') \probsmodels \varphi$.

%%%%%%%%%%%%%%%  public-history semantics

\paragraph*{Public-history semantics:}
In the public-history semantics all agents $a$ have full information
about the history $h$, i.e., $\Obsset_a=P$ 
and 
$\obs_a(h)=h$. Then,
  $(f,h) \publicmodels \StMod{a}{\varphi}$ iff 
     there exist 
     a state $t\in \Reach(\cT_a,h)$ 
     and a trace
     $f'\in \Traces^P(\cT_a,t)$ such that 
     $\last(h)=\first(f')$ and
     $(f',h) \publicmodels \varphi$.

%%%%%%%%%%%%  decremental semantics

\paragraph*{Decremental semantics:}
\label{semantics-decremental-obs-based}
The decremental semantics is a variant of the pure observation-based 
semantics with a different meaning of nested standpoint subformulas:
A standpoint formula $\StMod{a}{\varphi}$ 
interpretes
maximal standpoint subformulas $\StMod{b}{\psi}$ of $\varphi$ from the view of
agent $a$. The assumption is that $a$ 
knows agent $b$'s transition system $\cT_b$ and thus
can make the same guesses for the future as $b$, but can access
only the truth values of 
joint atomic propositions in $P_a' \cap P_b$ from the history to make a guess
what $b$ has observed in the past.
Here, $P_a' \subseteq P_a$ is the set of atomic propositions accessible
when interpreting $\StMod{a}{\varphi}$ 
which can itself be a subformula
of a larger standpoint formula $\StMod{c}{\chi}$ 
(in which case $\StMod{a}{\varphi}$ is interpreted from the view of agent $c$
and
$P_a' \subseteq P_c \cap P_a$).
Formally, we use a parametric satisfaction relation
$\decrmodels{Q}$ for $Q \subseteq P$.
The intuitive meaning of $\StMod{a}{\varphi}$ under $\decrmodels{Q}$ is that  $a$ can extract from the history 
$h$ exactly the truth values of the propositions in 
$\Obsset_a^Q= Q \cap P_a$. So,
$(f,h) \decrmodels{Q} \StMod{a}{\varphi}$ iff
     there exist a word $h' \in (2^{P})^+$, 
     a state 
     $t\in \Reach(\cT_a,h')$ 
     and a trace
     $f'\in \Traces^P(\cT_a,t)$ such that 
     $\last(h')=\first(f')$,
     $\proj{h}{Q \cap P_a}= \proj{h'}{Q \cap P_a}$
     and $(f',h') \decrmodels{Q \cap P_a} \varphi$.
For the satisfaction over a $\StpLTL$ structure, we start with $Q=P$.

%%%%%%%%%%%%%%%%  incremental semantics

\paragraph*{Incremental semantics:}
%`
The incremental semantics also relies on a parametric satisfaction relation
$\incrmodels{Q}$ where $Q \subseteq P$.
The intuitive meaning of 
$\StMod{a}{\varphi}$ 
under $\incrmodels{Q}$ 
is 
that agent $a$ can extract from the history $h$ exactly the propositions
in $\Obsset_a^Q= Q \cup P_a$ and interpretes $\varphi$
over $Q \cup P_a$.
Thus, when interpreting nested standpoint subformulas
$\StMod{b}{\psi}$ of $\StMod{a}{\varphi}$  then agents $a$ and $b$ build a
coalition to extract the information from the history.
Formally, $(f,h) \incrmodels{Q} \StMod{a}{\varphi}$ iff
     there exist a word $h' \in (2^{P})^+$, 
     a state 
     $t\in \Reach(\cT_a,h')$ 
     and 
     a trace
     $f'\in \Traces^P(\cT_a,t)$ such that 
     $\last(h')=\first(f')$,
     $\proj{h}{Q \cup P_a}= \proj{h'}{Q \cup P_a}$
     and $(f',h') \incrmodels{Q \cup P_a} \varphi$.
For the satisfaction over a $\StpLTL$ structure, we start with $Q=\varnothing$.

\begin{example}
	\normalfont
Let us  illustrate  the differences between the semantics:
Consider two agents $a$ and $b$ with $P_a=\{p,r\}$, $P_b=\{q,r\}$, and the following standpoint transition systems:
\begin{center}
\begin{tikzpicture}[->,>=stealth',shorten >=1pt,auto, semithick]
				\tikzstyle{every state} = [text = black]
				
				%left
				\node[state,scale=0.6] (lroot1) [inner sep = 0pt, fill = yellow!50] {$\{p\}$};
				\node[state, scale=0.6] (l1) [below of = lroot1, node distance = 1.35cm, fill = magenta!50] {$\{r\}$};
				\node[scale = 0.6] (linit1) [left of = lroot1, node distance = 1.2cm] {};
				
				\node[state,scale=0.6] (lroot2) [inner sep = 0pt, right of = lroot1, node distance = 2.5cm] {$\emptyset$};
				\node[state, scale=0.6] (l2) [below of = lroot2, node distance = 1.35cm] {$\emptyset$};
				\node[scale = 0.6] (linit2) [left of = lroot2, node distance = 1.2cm] {};
				
				\path 
					(linit1) edge (lroot1)
					(lroot1) edge (l1)
					(l1) edge [loop left] (l1)
					(linit2) edge (lroot2)
					(lroot2) edge (l2)
					(l2) edge [loop right] (l2)
				
				;

				%%%right 
				\node[state,scale=0.6] (rroot1) [inner sep = 0pt, fill = cyan!50, right of = lroot2, node distance = 5cm] {$\{q\}$};
				\node[state, scale=0.6] (r1) [below of = rroot1, node distance = 1.35cm] {$\emptyset$};
				\node[scale = 0.6] (rinit1) [left of = rroot1, node distance = 1.2cm] {};
				
				\node[state,scale=0.6] (rroot2) [inner sep = 0pt, right of = rroot1, node distance = 2.5cm] {$\emptyset$};
				\node[state, scale=0.6] (r2) [below of = rroot2, node distance = 1.35cm, fill = magenta!50] {$\{r\}$};
				\node[scale = 0.6] (rinit2) [left of = rroot2, node distance = 1.2cm] {};
				
				\path 
				(rinit1) edge (rroot1)
				(rroot1) edge (r1)
				(r1) edge [loop left] (r1)
				(rinit2) edge (rroot2)
				(rroot2) edge (r2)
				(r2) edge [loop right] (r2)
				
				;
				
				\node [left of = linit1, node distance = 0.35cm, scale = 0.8] {$\cT_{a}$};
				\node [left of = rinit1, node distance = 0.35cm, scale = 0.8] {$\cT_{b}$};
				
			\end{tikzpicture}
\end{center}
We consider  formulas $\phi=\StMod{a}{\neXt \StMod{b}{x}}$ with $x\in P = P_a \cup P_b$.
As these formulas contain only the existential $\StMod{}{}$ modality, satisfaction becomes  easier if less restrictive  requirements are imposed  on histories.
So, for all $\rho=(f,h)$ and all $x$ as above:
$
\rho \publicmodels \phi \Rightarrow \rho \incrmodels{} \phi \Rightarrow \rho \probsmodels \phi \Rightarrow \rho \decrmodels{} \phi \Rightarrow \rho \stepmodels \phi$.

Now, we show that the reverse implications do not hold. Let $\rho=(f,h)$ with $h=\{p,q\}$ and $f=\{p,q\}\{r\}^\omega$.
For $\phi=\StMod{a}{\neXt \StMod{b}{r}}$, we have
$\rho \not\publicmodels \phi $: In the quantification of $\StMod{a}$, all possible future traces $f'$ have to start with $\{p,q\}$. So, according to $\cT_a$, the future has to contain $r$ in the second step.
When evaluating $\StMod{b}$, these first two steps have to be respected, i.e., the history under consideration is $h'=\{p,q\} X$ where $X$ contains $r$. But there is no trace in $\cT_b$ with $q$ in the first and $r$ in the second step.
However, $\rho \incrmodels{} \phi $ as we can choose any trace $\{p\}\{r\}...$ as the future trace in the quantification of $\StMod{a}$. When projecting to $P_b$ we obtain a trace in $\cT_b$ satisfying $r$ at the second step.

Next, consider $\phi=\StMod{a}{\neXt \StMod{b}{p}}$ and $\rho$ as above.
Then, $\rho \not\incrmodels{} \phi $ as the future trace $f'$ quantified in $ \StMod{a}$ cannot have $p$ in the second step because there is no such trace in $\cT_a$. In the incremental semantics, the history chosen for $\StMod{b}$ 
has to agree with $f'$  on $p$ in the second step and hence can also not contain $p$ there.
In contrast, $\rho \probsmodels \phi $ as here agent $b$ can ``choose''   the truth value of $p$ arbitrarily as $p\not \in P_b$.

With similar reasoning, $\probsmodels$ and $\decrmodels{}$ can be distinguished by the formula $\phi=\StMod{a}{\neXt \StMod{b}{(q\land r)}}$ and $\rho$ as above:
$\rho \not \probsmodels \phi$ as $q$ is  false on any state reachable in $\cT_b$ after one step. So, no matter which history $h'$ and future $f'$ is chosen in the quantification of
$ \StMod{a}$, there is no  future-history pair $(f'',h'')$ that agrees with a trace in $\cT_b$ on all atomic propositions in $P_b$ and is such that $q$ holds at the second position. 
In contrast, $\rho \decrmodels{} \phi$: We can choose $h'=\{p\}$ and $f'=\{p\}\{r\}^\omega$ for the quantification of $ \StMod{a}$ as $h'$ agrees with $h=\{p,q\}$ on all atomic propositions in $P_a$. Then, the quantification of $ \StMod{b}$ in $\phi$ can choose $h''=\emptyset \{r,q\}^\omega$ and $f''=\{r,q\}^\omega$ as this agrees with the trace of the right-hand side path in $\cT_b$ for all atomic propositions in $P_a\cap P_b = \{r\}$.

Finally, for $\decrmodels{}$ and  $\stepmodels$ consider  $\phi=\StMod{a}{\neXt \StMod{b}{\neg r}}$ and $\rho$ as above.
Under $\decrmodels{}$, in the first quantification of $ \StMod{a}$, the history $h'$ has to contain $p$ and potentially $q$ and so the future $f'$ has to contain $r$ in the second step as only the left initial state in $\cT_a$ is possible. So, the history $h''$ quantified in $ \StMod{b}$ also has to contain $r$ in the second step as $r\in P_a\cap P_b$ and so $\phi$ cannot hold, i.e., $\rho \not \decrmodels{} \phi$.
In the step semantics, however, the quantification for $\StMod{b}$ can simply choose history  $h''=\{q\}\emptyset$ and future  $f''=\emptyset^\omega$ and so $\rho \stepmodels \phi$. 
\end{example}

%%%%%%%%%%%%%%%%%%%%%%%%%%%%%%%%%%%%%%%%%%%%%%%%%%%%%%%%%%%%%%%%%%%%%%%%%

\paragraph*{Satisfaction of SLTL formulas over structures.}
Given a structure $\mathfrak{T}= \bigl(\cT_0,(\cT_a)_{a\in \Ag}\bigr)$
as in Section \ref{sec:structures}, an $\StpLTL$ formula $\phi$
and 
${\models}  \in \{{\stepmodels}, 
  {\probsmodels}, {\publicmodels}, {\decrmodels{}}, {\incrmodels{}} \}$:
\begin{center}
  $\mathfrak{T}\models \phi$
  \quad iff \quad
  $(f,\first(f)) \models_* \phi$   
  for all $f \in \Traces(\cT_0)$
\end{center}
where 
$\models_*$ equals $\models$ for the step, pure observation-based
and public-history semantics,
$\models_*$ equals $\decrmodels{P}$ for the decremental semantics
and $\models_*$ equals $\incrmodels{\varnothing}$ 
for the incremental semantics.

%%%%%%%%%%%%%%%%%%%%%%%%%%%%%%%%%%%%%%%%%%%%%%%%%%%%%%%%%%%%%%%%%%%%%%%%%

\begin{remark}
\label{remark:context-dependency-syntax-tree}
{\rm
The meanings of $\StMod{a}{\varphi}$ under 
$\decrmodels{Q}$ and $\incrmodels{Q}$ are context-dependent through
the parameter $Q$.
If, e.g.,  
$\phi = \StMod{a}{\varphi} \vee \StMod{b}{(q \wedge \neXt \StMod{a}{\varphi})}$ for $\varphi$ an $\StpLTL$ formula
then the first occurrence of $\StMod{a}{\varphi}$ 
is interpreted over $\decrmodels{P}$, 
while the second is interpreted over $\decrmodels{P_b}$.
The set $Q = Q_{\chi}^*$ with $*\in \{\decr,\incr\}$ 
over which occurrences of subformulas $\chi$ 
of $\phi$ are interpreted
can be derived from the syntax tree of $\phi$.

Formally, we assign a set $Q_v^*$ with $*\in \{\decr,\incr\}$ 
to each node $v$ in the tree. Let $\phi_v$ denote the subformula represented by $v$.
If $v$ is the root node then $\phi_v=\phi$ and
$Q_{v}^{\tinydecr}=P$ while
$Q_{v}^{\tinyincr}=\varnothing$.
Let now $v$ and $w$ be nodes such that $v$ 
is the father of $w$.
If $v$ is not labeled by a standpoint modality then
$Q_v^*=Q_w^*$. Otherwise $\phi_v$ has the form $\StMod{a}{\phi_w}$
for some $a\in \Ag$
and $Q_w^{\tinydecr} = Q_v^{\tinydecr} \cap P_a$ while 
$Q_w^{\tinyincr} = Q_v^{\tinyincr} \cup P_a$.
With abuse of notations,
we will simply write $Q_{\chi}^*$ instead of $Q_v^*$ for a node $v$ 
representing the particular occurrence of $\chi = \phi_v$ in $\phi$.
  }
\end{remark}

%%%%%%%%%%%%%%%%%%%%%%%%%%%%%%%%%%%%%%%%%%%%%%%%%%%%%%%%%%%%%%%%%%%%%%%%%

Obviously, the five semantics agree on the LTL fragment.
More precisely,
for $\models$ as above,  
$\fT$ a $\StpLTL$ structure and $\phi$ an LTL formula
we have 
$\fT \models \phi$ iff $\cT_0 \LTLmodels \phi$ in the sense that
$\Traces(\cT_0) \subseteq \{ f \in (2^P)^{\omega} : f \LTLmodels \phi\}$. 
For $\StpLTL$ formulas
with at least one standpoint subformula, but no alternation
of standpoint modalities,
the pure observation-based, decremental
and incremental semantics agree,
but might yield different truth values than the step or public-history 
semantics:

\begin{lemma}
\label{lemma:pobs=decr=incr-for-alternation-depth-1}
If $\phi$ is a $\StpLTL$ formula of alternation depth $1$ then
$\fT \pobsmodels \phi$ 
iff $\fT \decrmodels{} \phi$ 
iff $\fT \incrmodels{} \phi$, while
$\fT \not\pobsmodels \phi$ 
and $\fT \models \phi$ 
(or vice versa)
is possible where ${\models} \in \{\stepmodels,\publicmodels\}$.
Likewise,
$\fT \not\publicmodels \phi$ 
and $\fT \stepmodels \phi$ (or vice versa) is possible.
\end{lemma}

%%%%%%%%%%%%%%%%%%%%%%%%%%%%%%%%%%%%%%%%%%%%%%%%%%%%%%%%%%%%%%%%%%

\subsection{From pure observation-based $\StpLTL$ to LTLK}

\label{sec:embedding-StpLTL-into-LTLK}

LTLK (also called CKL$_m$) 
\cite{HalVar-STOC86,HalVar-JCSS89,MC-LTLK-and-beyond-2024} 
is an extension of LTL by an unary knowledge modality $K_a$ for 
every agent $a \in \Ag$ and a common knowledge operator $C_A$ for 
coalitions $A \subseteq \Ag$.
We drop the latter and deal with LTLK where the grammar 
for formulas is the same as for $\StpLTL$ when
$\StMod{a}{\varphi}$ is replaced with $K_a \varphi$.
The alternation depth of LTLK formulas and the sublogics LTLK$_d$ are defined as for $\StpLTL$.

LTLK structures are tuples 
$(\cT,(\sim_a)_{a\in \Ag})$ where 
$\cT = (S,\to,\init,R,L)$ is a transition system and the
$\sim_a$, $a \in \Ag$, are equivalence relations on $S$. The intended meaning
is that $s \sim_a t$ if agent $a$ cannot distinguish the states $s$ and $t$.

The perfect-recall semantics of LTLK extends
the standard LTL semantics formulated for path-position pairs $(\pi,n)$ 
consisting of a path $\pi = s_0 s_1 s_2 \ldots \in \Paths(\cT)$ 
and a position $n \in \Nat$
by $(\pi,n) \LTLKmodels K_a\varphi$ iff 
$(\pi',n) \LTLKmodels \varphi$ for all paths 
$\pi' = s_0' s_1' s_2' \ldots \in \Paths(\cT)$ with
$s_i \sim_a s_i'$ for $i=0,1,\ldots ,n$.
As such, the dual standpoint modality
$\DualStMod{a}{}$ under the pure observation-based semantics
resembles the $K_a$ modality of LTLK.
Indeed, $\StpLTL$ under $\probsmodels$ can be considered
as a special case of LTLK under the perfect-recall semantics:

\begin{lemma}
\label{lemma:embdding-pobs-in-LTLK}
\label{lemma:embedding-pobs-in-LTLK}
  Given a pair $(\fT,\phi)$ consisting of a $\StpLTL$ structure 
 $\fT = (\cT_0, (\cT_a)_{a\in \Ag})$ and a $\StpLTL$ formula $\phi$, 
 one can construct in polynomial time
 an LTLK structure $\ltlk{\fT} = (\cT',(\sim_a)_{a\in \Ag})$
 and an LTLK formula $\ltlk{\phi}$ such that:
\begin{enumerate}
\item [(1)] 
   $\fT \pobsmodels \phi$ if and only if 
   $\ltlk{\fT} \LTLKmodels \ltlk{\phi}$ 
\item [(2)] 
   $\phi$ and $\ltlk{\phi}$ have the same alternation depth.
\end{enumerate}
\end{lemma}

\begin{proofsketch}
	Assume w.l.o.g. that the initial state labelings in $\cT_0$ and $\cT_a$ are consistent for all $a$. 
   First, for each $a \in \Ag \cup \{0\}$, we extend $\cT_a$ to its completion $\cT_a^{\bot}$ by adding new states $\bot_a^Q$ for each $Q \subseteq P_a$, where the set of atomic propositions in $\cT_a^{\bot}$ is 
	$P_a^{\bot}=P_a \cup \{\bot_a\}$.
	The labeling of the original states is unchanged, i.e.,
	$L_a^{\bot}(s)=L_a(s)$, while
	for the fresh states $L_a^{\bot}(\bot_a^Q)=Q \cup \{\bot_a\}$. 
   We then construct the synchronous product $\cT' = S_0^{\bot}\times\prod_{a \in \Ag} \cT_a^{\bot}$	such that $L_0^{\bot}(s) \cap P_a = L_a^{\bot}(s_a)\cap P_a$ for each $a \in \Ag$.
	The set of atomic propositions is 
	$P'=P \cup \{\bot_a : a \in \Ag \cup \{0\}\}$, and 
	$L'(s, (s_a)_{a\in \Ag}) = L_0(s) \cup 
	\{\bot_a : a \in \Ag \cup \{0\}, \bot_a \in L_a^{\bot}(s_a)  \}$.
	The initial state of $\cT'$ is 
	$\init' = (\init_0, (\init_a)_{a\in \Ag})$, and the transitions are defined synchronously. The equivalence relations $\sim_a$ on $S'$ satisfy $\sigma \sim_a \theta$ iff $L'(\sigma) \cap P_a^{\bot} = L'(\theta) \cap P_a^{\bot}$.
	We then inductively
	translate $\StpLTL$ formulas $\varphi$ into ``equivalent'' LTLK formulas $\primed{\varphi}$:
	$
	\primed{\varphi} = \varphi, \ \primed{(\neg \varphi)} = \neg \primed{\varphi}, \  \primed{(\varphi \wedge \psi)} = \primed{\varphi} \wedge \primed{\psi}, \ \primed{(\neXt \varphi)} = \neXt\primed{\varphi}, \ \primed{(\varphi_1 \Until \varphi_2)} = \primed{\varphi_1} \! \! \! \Until \primed{\varphi_2}, \ \primed{(\StMod{a}{\varphi})} = \overline{K}_a (\primed{\varphi} \wedge \Box \neg \bot_a).
	$
	Finally, we define $\ltlk{\phi} = \Box \neg \bot_0 \to \primed{\phi}$.
\end{proofsketch}

%% file: Sections/model-checking.tex
%%%%%%%%%%%%%%%%%%%%%%%%%%%%%%%%%%%%%%%%%%%%%%%%%%%%%%%%%%%%%%%%%%%%%%%%%
%%%%%%%%%%%%%%%%%%%%%%%%%%%%%%%%%%%%%%%%%%%%%%%%%%%%%%%%%%%%%%%%%%%%%%%%%
%%%%%
%%%%%     section:  model checking
%%%%%
%%%%%%%%%%%%%%%%%%%%%%%%%%%%%%%%%%%%%%%%%%%%%%%%%%%%%%%%%%%%%%%%%%%%%%%%%
%%%%%%%%%%%%%%%%%%%%%%%%%%%%%%%%%%%%%%%%%%%%%%%%%%%%%%%%%%%%%%%%%%%%%%%%%

\section{Model checking $\StpLTL$}

\label{sec:model-checking}

Let $\models$ be one of the five satisfaction relations defined in
Section \ref{sec:semantics}. The task of the $\StpLTL$ model checking problem
is as follows. Given a $\StpLTL$ structure 
$\fT = \bigl(\cT_0,(\cT_a)_{a\in \Ag}\bigr)$
and an $\StpLTL$ formula $\phi$, decide whether
$\fT \models \phi$.

We present a generic model checking algorithm for the five semantics,
which can be seen as an adaption of the standard CTL*-like model checking 
procedure \cite{EmerSistla-CTLstar-MC-InfCon84,Emerson-Lei-87} 
and its variant for LTLK and CTL*K under the perfect-recall semantics
\cite{MC-LTLK-and-beyond-2024}.
The essential idea is an inductive approach to compute 
deterministic finite automata (DFA)
$\cD_\chi$ for all standpoint subformulas $\chi = \StMod{a}{\varphi}$
of $\phi$. 
The automaton for $\chi$, called history-DFA, represents the language 
of all histories $h \in (2^P)^+$ with
$(*,h) \InnerModels{\chi} \chi$
where $\InnerModels{\chi}$ stands for one of the satisfaction
relations
$\stepmodels$, $\pobsmodels$, $\publicmodels$,
$\decrmodels{Q}$ or $\incrmodels{Q}$ 
with $Q \in \{ Q_{\chi}^{\tinydecr}, Q_{\chi}^{\tinyincr}\}$ as in Remark \ref{remark:context-dependency-syntax-tree}.

The inductive computation corresponds to a bottom-up approach where the innermost standpoint  subformulas, i.e., 
subformulas $\chi = \StMod{a}{\psi}$ where 
$\psi$ is an LTL formula, are treated first.
When computing $\cD_{\chi}$ for
a subformula $\chi=\StMod{a}{\psi}$ where $\psi$ contains 
standpoint modalities,
we may assume that history-DFA 
for the maximal standpoint subformulas $\chi_1,\ldots,\chi_k$
of $\psi$ have been computed before.

%%%%%%%%%%%%%%%%%%%%%%%%%%%%%%%%%%%%%%%%%%%%%%%%%%%%

\paragraph{Preprocessing for $\decrmodels{}$ and $\incrmodels{}$.}
When dealing with the decremental and incremental semantics,
our iterative bottom-up approach requires the computation of the relevant parameters $Q \subseteq P$ of 
the satisfaction relations $\decrmodels{Q}$ and $\incrmodels{Q}$
over which subformulas of $\phi$ are interpreted.
This can be done in time linear in the length of $\phi$ by a top-down analysis
of the syntax tree of $\phi$ 
(see Remark \ref{remark:context-dependency-syntax-tree}).
In what follows, we shall write $Q_{\chi}$ to denote
the corresponding subset $Q_{\chi}^{\tinydecr}$ or $Q_{\chi}^{\tinyincr}$, respectively,
of $P$ for (an occurrence of) a
subformula $\chi$. 
\optional{
Thus, if $\chi$ is 
not in the scope of a standpoint modality
or a maximal standpoint subformulas of $\phi$ then
$Q_{\chi} =P$ under the decremental semantics and
$Q_{\chi} =\varnothing$ under the incremental semantics.
If $\chi= \StMod{a}{\psi}$ is (an occurrence) of a 
standpoint subformula of $\phi$ and
$\chi'$ a subformula of $\psi$ that -- as a subformula of $\psi$ -- is 
not in the scope of a standpoint modality or
a maximal standpoint subformula of $\psi$ then
$Q_{\chi'} =Q_{\chi} \cap P_a$ for the decremental semantics
and $Q_{\chi'}=Q_{\chi}\cup P_a$ for the incremental semantics.}

\paragraph{\bf Observation sets $\Obsset_{\chi}$.}
Our algorithm 
uses subsets $\Obsset_{\chi}\subseteq P$ for the standpoint subformulas
$\chi = \StMod{a}{\varphi}$ of $\phi$, whose
definition 
depends on the considered semantics:
$\Obsset_{\chi} = \varnothing$ for $\stepmodels$, 
$\Obsset_{\chi} = P_a$ for $\probsmodels$,
$\Obsset_{\chi} = P$ for $\publicmodels$,
$\Obsset_{\chi} = Q_{\chi}^{\tinydecr} \cap P_a$ for $\decrmodels{}$
and $\Obsset_{\chi} = Q_{\chi}^{\tinyincr} \cup P_a$ for $\incrmodels{}$.
Let $\obs_{\chi} : (2^P)^* \to (2^{\Obsset_{\chi}})^*$ 
be the induced observation function
$\obs_{\chi}(h)=\proj{h}{\Obsset_{\chi}}$.

\paragraph{Transition systems $\cT_a^R$.}
From $\Obsset_{\chi}$ we will derive a superset $R = R_{\chi} = P_a \cup \Obsset_{\chi}$ of $P_a$
and switch from $\cT_a$ to a transition system $\cT_a^R$ over $R$ 
that behaves as
$\cT_a$, but makes nondeterministic guesses for the truth values
of propositions in $R \setminus P_a$ for the starting state, as well as
in every step of a computation.
So, the states of $\cT_a^R$ are pairs
$(s,O)$ with $s \in S_a$ and $O \subseteq R$ such that
$L_a(s) = O \cap P_a$.
For the precise definition of $\cT_a^R$, see
Definition \ref{def:T-a-R} in the appendix.
The switch from $\cT_a$ to $\cT_a^R$ 
will be needed to deal with the constraints
$\last(h')=\first(f')$ in the definitions of
$\publicmodels{}$ and $\incrmodels{}$.

%%%%%%%%%%%%%%%%%%%%%%%%%%%%%%%%%%%%%%%%%%%%%%%%%%%%%%%%%%%%%%%%%%%%%%%%%%%%%

\begin{remark}
\label{T-a-R}
{\rm 
If $R =P_a$ then $\cT_a^R$ and $\cT_a$ are isomorphic, so the switch from
$\cT_a$ to $\cT_a^R$ is obsolete. (Note that then $O=L_a(s)$ for all states $(s,O)$ in $\cT_a^R$.)
This applies to $\stepmodels$ where $\Obsset_{\chi}=\varnothing$,
$\probsmodels$ where $\Obsset_{\chi}= P_a$,
and $\decrmodels{}$ where 
$\Obsset_{\chi} \subseteq P_a$.
  }
\end{remark}

%%%%%%%%%%%%%%%%%%%  definition of history DFA  

\begin{definition}
 \label{def:history-DFA}
  Let $\chi$ be a standpoint subformula of $\phi$.
  A history-DFA for $\chi$ is a DFA $\cD$ over the alphabet $2^P$
  such that 
  \begin{enumerate}
  \item [(1)] 
        The language of $\cD$ is
        $\{h \in (2^P)^+ : (*,h) \InnerModels{\chi} \chi \}$. 
  \item [(2)] 
        $\cD$ is $\Obsset_{\chi}$-deterministic 
        in the following sense:
        Whenever $h_1,h_2 \in (2^P)^+$ such that
        $\proj{h_1}{\Obsset_{\chi}}=\proj{h_2}{\Obsset_{\chi}}$ then 
        $\delta_{\chi}(\init_{\chi},h_1)=\delta_{\chi}(\init_{\chi},h_2)$, where  $\delta_{\chi}$ denotes the transition relation of $\cD$.
   \end{enumerate}
\end{definition}

%%%%%%%%%%%%%%%%%%%%%%%%%%%%%%%%%%%%%%%%%%%%%%%%%%%%%%%%%%%%%%%%%%%%%%%

\subsubsection*{Generic $\StpLTL$ model checking algorithm}

%%%%%%%%%%%   basis of induction

\paragraph{Basis of induction.}
In the basis of induction we are given a standpoint formula
$\chi = \StMod{a}{\varphi}$ with $\varphi$ an LTL formula.
Let $\cT_{\chi} = \cT_a^R$ where $R = R_{\chi}= P_a \cup \Obsset_{\chi}$.
We apply a mild variant of standard LTL model checking techniques 
\cite{VarWolp-LICS86,VarWolp-InfComp94,BK08,CGKPV18}
(see Section \ref{sec:LTL-MC} in the appendix)
to compute the set $\Sat_{\cT_a^R}(\exists \varphi)$ 
of states $(s,O)$ in $\cT_a^R$ for which 
there is a $f \in \Traces^P(\cT_a^R,(s,O))$ with 
$f \LTLmodels \varphi$.
We then use a powerset construction
applied to $\cT_a^{R}$ (see Definition \ref{def:default-history-DFA} below) to construct an $\Obsset_{\chi}$-deterministic 
history-DFA $\cD_{\chi}$.

%%%%%%%%%%%%%%%%%%%%%%   step of induction

\paragraph{Step of induction.}
Suppose now 
$\chi = \StMod{a}{\varphi}$ and $\varphi$ has standpoint subformulas.
Let $\chi_1,\ldots,\chi_k$ be the maximal standpoint subformulas of $\varphi$,
say $\chi_i = \StMod{b_i}{\psi_i}$ for $i=1,\ldots,k$, and let
$\cD_1=\cD_{\chi_1}, \ldots, \cD_k=\cD_{\chi_k}$ be their history-DFAs 
over
$2^P$.
Let $\cD_i = (X_i,\delta_i,\init_i,F_i)$ and $\Obsset_i = \Obsset_{\chi_i}$. Further, let $R_{\chi} = P_a \cup \Obsset_{\chi}$.
We consider the transition system
\begin{center}
  $\cT_{\chi} = \cT_a^{R_{\chi}} \bowtie \cD_1 \bowtie \ldots \bowtie \cD_k$
\end{center}
that is obtained by putting $\cT_a^{R_{\chi}}$ in parallel to
the product of the $\cD_i$'s. 
For this, we introduce pairwise distinct, fresh atomic propositions
$p_1,\ldots,p_k$ for each of the $\chi_i$'s and define the components
$\cT_{\chi} = (Z_{\chi},\to_{\chi},\Init_{\chi},\cP_{\chi},L_{\chi})$ 
as follows (where we write $R$ instead of $R_{\chi}$).
  The state space is 
  $Z_{\chi} \ = \ S_a^{R} \times X_1 \times \ldots \times X_k$.
  The transition relation satisfies:
  $((s,O),x_1,\ldots,x_k) \to_{\chi} ((s',O'),x_1',\ldots,x_k')$
  iff there are $H, H' \subseteq P$ such that
  $H \cap \Obsset_{\chi} = H' \cap \Obsset_{\chi}$, 
  $H \cap R=O$, $(s,O) \to_a^{R} (s',O')$  
  and
  $x_i'=\delta_i(x_i,H')$ for $i=1,\ldots,k$.
  The set of atomic propositions is $\cP_{\chi} = R \cup \{p_1,\ldots,p_k\}$.
  The labeling function $L_{\chi} : Z_{\chi} \to 2^{\cP_{\chi}}$ is given by
  $L_{\chi}((s,O),x_1,\ldots,x_k) \cap R = O$ 
  and 
  $p_i\in L_{\chi}(s,x_1,\ldots,x_k)$ iff $x_i \in F_i$ for $i=1,\ldots,k$.
  Lastly, $\Init_{\chi}$ contains all states
  $((\init_a,O),x_1,\ldots,x_k)$ such that 
  $(\init_a,O)\in \Init_a^R$  (i.e., $O\cap P_a = L_a(\init_a)$)
  and
  there exists $H' \subseteq P$
  with $x_i=\delta_i(\init_i,H')$ for $i=1,\ldots,k$ and
  $O\cap \Obsset_{\chi}= H'\cap \Obsset_{\chi}$.

We now replace $\varphi$ with the LTL formula 
$\varphi' \ = \ \varphi[\chi_1/p_1,\ldots,\chi_k/p_k]$
over $\cP=P \cup \{p_1,\ldots,p_k\}$ that results
from $\varphi$ by syntactically 
replacing the maximal standpoint subformulas $\chi_i$ of $\varphi$ with $p_i$.
We then apply standard LTL model checking techniques 
(see Section \ref{sec:LTL-MC} in the appendix)
to compute $\Sat_{\cT_{\chi}}(\exists \varphi') = \{ z \in Z_{\chi} : \exists f \in 
\Traces^{\cP}(\cT_{\chi},z) \text{ s.t. } f \LTLmodels \varphi' \}$.

Having computed $\Sat_{\cT_{\chi}}(\exists \varphi')$,  
the atomic propositions $p_1,\ldots,p_k$ are no longer needed.
We therefore switch from $\cT_{\chi}$ to
$\cT_{\chi}'= \proj{\cT_{\chi}}{R}$ 
which agrees with $\cT_{\chi}$ except that $\cP_{\chi}' = R$ and 
$L_{\chi}'(s,x_1,\ldots,x_k)= L_{\chi}(s,x_1,\ldots,x_k) \cap R$.
By applying a powerset construction to $\cT_{\chi}'$ we get
an $\Obsset_{\chi}$-deterministic history DFA $\cD_{\chi}$
for the language $\{h \in (2^P)^+ : (*,h)\InnerModels{\chi} \chi\}$
(see Definition \ref{def:default-history-DFA} below, 
applied to $\cT=\cT_{\chi}'$,
$\Obsset_{\chi}$ and 
the set $U_{\chi}=\Sat_{\cT_{\chi}}(\exists \varphi')$ for the declaration
of the final states in $\cD_{\chi}$).

%%%%%%%%%%  final step: evaluation over StpLTL structure

\paragraph{Final step.}
After treating all maximal standpoint subformulas of $\phi$,
we need to check whether $\phi$ holds for all traces
of $\cT_0$. This is done via similar techniques as in the first part 
of the step of induction, i.e., we build the product $\cT_{\phi}$ 
of $\cT_0$ with the history-DFAs that have been
constructed for the maximal standpoint subformulas of $\phi$ and
replace them with fresh atomic propositions, yielding an LTL formula
$\phi'$ over an extension of $P$. We then apply
standard LTL model checking techniques to check whether all traces
of $\cT_{\phi}$ satisfy $\phi'$. 
\optional{%
This can be done using an algorithm that is polynomially space-bounded 
in the size of $\cT_{\phi}$ and the length of $\phi'$ (and $\phi$).}

%%%%%%%%%%%%%%%%%%%%%%%%%%%%%%%%%%%%%%%%%%%%%%%%%%%%%%%%%%%%%%%%%%%%%%%%%
%%%%%%
%%%%%%    history-DFA via powerset construction
%%%%%%
%%%%%%%%%%%%%%%%%%%%%%%%%%%%%%%%%%%%%%%%%%%%%%%%%%%%%%%%%%%%%%%%%%%%%%%%%

\subsubsection*{Computation of history-DFA}

\label{sec:default-history-DFA}

The history-DFA $\cD_{\chi}$ 
for standpoint subformulas $\chi=\StMod{a}{\varphi}$ 
can be obtained by applying a
powerset construction to the transition system 
$\cT=\cT_{\chi}$ over $R=R_{\chi}$.
This construction is an adaption of the classical powerset construction
for nondeterministic finite automata and can be seen as a one-agent variant
of the powerset construction introduced by Reif \cite{Reif-JCSS84}
for partial-information
two-player games and its variant for CTL*K model checking 
\cite{MC-LTLK-and-beyond-2024}.
The acceptance condition in $\cD_{\chi}$ is derived from  
the set $U = \Sat_{\cT}(\exists \varphi)$ in the basis of induction and
$U = \Sat_{\cT}(\exists \varphi')$ in the step of induction.

%%%%%%%%%%%  definition default history-DFA

\begin{definition}
\label{def:default-history-DFA}
Let $\cT = (S,\to,\Init,R,L)$ 
with $R \subseteq P$,
$\Obsset \subseteq P$ an observation set,
and $U \subseteq S$.
Then,
$\pow(\cT,\Obsset,U)$
is the following
DFA $\cD= (\cX_{\cD},\delta_{\cD},\init_{\cD},F_{\cD})$ over the alphabet $2^P$.
   The state space is $\cX_{\cD} = 2^S \cup \{\init_{\cD}\}$.
  The transition function $\delta_{\cD} : \cX_{\cD} \times 2^P \to \cX_{\cD}$ is 
  given by (where $x \in 2^S$):
  \begin{center}
  \begin{tabular}{l}
  \begin{tabular}{r}
    $\delta_{\cD}(x,H)$  =          
    $\bigr\{ s' \in S : \text{ there exists } s \in x \text{ with }$
      \ \ \\
        
        $s \to s'$ \text{ and } $L(s') \cap \Obsset = H \cap \Obsset \bigr\}$
    \\[1ex]
   \end{tabular}   
   \\
   \begin{tabular}{l}
   $\delta_{\cD}(\init_{\cD},H)$  =   
        $\bigl\{ s \in \Init :  L(s) \cap \Obsset = H \cap \Obsset 
        \bigr\}$
   \end{tabular}
   \end{tabular}
  \end{center}
  The set of final states is
  $F_{\cD}= 
   \bigl\{  x \in \cX_{\cD} :  x \cap U \not= \varnothing  \bigr\}$.
\end{definition}

%%%%%%%%%%%%%%%%%%%%%%%%%%%%%%%%%%%%%%%%%%%%%%%%%%%%%%%%%%%%%%%%%%%

For the soundness proof, see
Section \ref{sec:soundness-history-DFA} in the Appendix.

%% file: Sections/complexity.tex
%%%%%%%%%%%%%%%%%%%%%%%%%%%%%%%%%%%%%%%%%%%%%%%%%%%%%%%%%%%%%%%%%%%%%
%%%%%%%%%%%%%%%%%%%%%%%%%%%%%%%%%%%%%%%%%%%%%%%%%%%%%%%%%%%%%%%%%%%%%
%%%%%
%%%%%    Complexity-theoretic results
%%%%%
%%%%%%%%%%%%%%%%%%%%%%%%%%%%%%%%%%%%%%%%%%%%%%%%%%%%%%%%%%%%%%%%%%%%%
%%%%%%%%%%%%%%%%%%%%%%%%%%%%%%%%%%%%%%%%%%%%%%%%%%%%%%%%%%%%%%%%%%%%%

\section{Complexity-theoretic results}

\label{sec:complexity}

At first glance the time complexity of
our generic model checking algorithm 
is  $d$-fold exponential for
$\StpLTL$ formulas $\varphi$ where $d$ is the maximal nesting depth of standpoint modalities in $\varphi$. 
This yields a nonelementary upper bound for the
$\StpLTL$ model checking problem.

For $*\in \{\step,\pobs,\public,\decr,\incr\}$, let $\StpLTL^*$ denote $\StpLTL$ under the semantics w.r.t. $\models_*$.
We now investigate the complexity 
of the different $\StpLTL^*$ model checking problems. 

%\optional{%
%As $\StpLTL$ under all 5 semantics is an extension of LTL 
%and the LTL model checking
%problem is known to be PSPACE-complete \cite{SistlaClarke-JACM85},
%the $\StpLTL^*$ model checking problem is PSPACE-hard.}

%%%%%%%%%%%%%%%%%%%%%%%%%%%%%%%%%%%%%%%%%%%%%%%%%%%%%%%%%%%%%%%%%%%%%%

\subsubsection*{Bounds on the sizes of history-DFA and
                time bounds}

Obviously, it suffices to construct the reachable fragment of the
default history-DFA. Furthermore, we may apply a 
standard poly-time minimization
algorithm to the default history-DFA. Thus, to establish better complexity
bounds it suffices to provide elementary 
bounds on the sizes of minimal history-DFA
for standpoint subformulas.
Let us start with a simple observation on how to exploit the property that
the history-DFA $\cD_{\chi}$ are $\Obsset_{\chi}$-deterministic.

\begin{remark}
\label{remark:exploit-obs-determinism}
{\rm
With the notations of the step of induction for $\chi=\StMod{a}{\varphi}$,
the history-DFA $\cD_{\chi}$ for $\chi$ is obtained by 
(the reachable fragment) of $\pow(\cT_{\chi}', \Obsset_{\chi}, U_{\chi})$ where 
$U_{\chi}=\Sat_{\cT_{\chi}}(\exists \varphi')$.
Recall that $\cT_{\chi}'$ has been defined via a product construction
$\cT_a^R \bowtie \cD_1 \bowtie \ldots \bowtie \cD_k$ where the $\cD_j$'s
are history-DFAs for the maximal subformulas $\chi_j = \StMod{b_j}{\psi_j}$ 
of $\varphi$.

Suppose now that $i\in \{1,\ldots,k\}$ such that $\cD_i$
is $\Obsset_{\chi}$-deterministic.
Then, for all histories $h \in (2^P)^+$  and states
$((s,O),x_{1},\ldots,x_{k})$ 
in $\Reach(\cT_{\chi}',h)$ we have
$x_i = \delta_i(\init_i,h)$.
Hence, if $x$ is a reachable state in $\cD$, 
say $x=\delta_{\cD_{\chi}}(\init_{\cD},h)$, then for all states
$((s,O),x_{1},\ldots,x_{k}) \in x$ we have 
$x_i = \delta_i(\init_i,h)$.
Thus, we may redefine the default history-DFA $\cD_{\chi}$ as a product 
$\pow(\cT_a^R \bowtie \prod_{j \not= i} \cD_j,\Obsset_{\chi}) 
    \bowtie \cD_i$ 
rather than a powerset construction of
$\cT_a^R \bowtie \cD_1 \bowtie \ldots \bowtie \cD_k$.
From now on, we write $\pow(\cT,\Obsset)$ for the
structure $(\cX_{\cD},\delta_{\cD},\init_{\cD})$ defined as in
Definition \ref{def:default-history-DFA}. The acceptance condition 
$F_{\cD_{\chi}}$
is (re)defined as the set of all states $(x,x_i)$ 
in the above product (i.e., $x$ is a subset of the state space of 
$\cT_a^R \bowtie \prod_{j \not= i} \cD_j$ 
and $x_i$ a state in $\cD_i$) such that 
$\{ (\xi,x_i) : \xi \in x\} \cap U_{\chi} \neq \varnothing$.
Let $I$ denote the set of indices $i \in \{1,\ldots,k\}$ where
the history-DFA for $\chi_i=\StMod{b_i}{\psi_i}$ 
is $\Obsset_{\chi}$-deterministic.
W.l.o.g. $I=\{1,\ldots,\ell\}$.
Then, we can think of $\cD_{\chi}$ as (the reachable
fragment of) a product
$\pow(\cT_a^R \bowtie \cD_{\ell+1} \bowtie \ldots \bowtie \cD_k, \Obsset_{\chi}) 
    \bowtie \cD_1 \bowtie \ldots \bowtie \cD_{\ell}$.
  }
\end{remark}

If $\chi=\StMod{a}{\varphi}$ then
$\Obsset_{\chi_i} = \Obsset_{\chi}$
for all maximal standpoint subformulas $\chi_i = \StMod{b_i}{\psi_i}$
of $\varphi$ where $a=b_i$.
Thus, by induction on $d=\ad(\chi)$,
Remark \ref{remark:exploit-obs-determinism}
yields that
the size of the reachable fragment of the default 
history-DFA for a standpoint formula
$\chi$ is $d$-fold exponentially bounded. Hence:

\begin{lemma}
\label{d-exp-time-bound}
  Under all five semantics,
  the time complexity of the 
  algorithm of Section \ref{sec:model-checking}
  is at most $d$-fold exponential when $d=\ad(\phi)$.
\end{lemma}

This $d$-fold exponential upper time bound will now be improved
using further consequences of 
Remark \ref{remark:exploit-obs-determinism}.

We start with $\stepmodels$.
Here, we have $\Obsset_{\chi} = \varnothing$ for all $\chi$.
With the notations used in Remark \ref{remark:exploit-obs-determinism},
the history-DFA $\cD_1,\ldots,\cD_k$ are
$\Obsset_{\chi}$-deterministic. Thus, we can think of
$\cD_{\chi}$ as a product 
$\pow(\cT_a,\varnothing) \bowtie \cD_1 \bowtie \ldots \bowtie \cD_k$.
But now, each of the $\cD_i$'s also has this shape.
In particular, if \( \chi_i = \StMod{a}{\psi_i} \), then the projection of the first coordinate of the states and the transitions between them in the reachable fragment of \( \cD_i \) matches exactly those of \( \pow(\cT_a, \varnothing) \).
Thus, one can incorporate the information on final states
and drop $\pow(\cT_a,\varnothing)$ from the product.
As a consequence, the default history-DFA $\cD_{\chi}$ 
can be redefined for $\stepmodels$ such that
the state space of $\cD_{\chi}$ is contained in
$\prod_{b\in \Ag(\chi)} 2^{S_b}$ 
where $\Ag(\chi)$ denotes the set of agents $b\in \Ag$ 
such that $\chi$ has a (possibly non-maximal) standpoint subformula
of the form $\StMod{b}{\psi}$. 
  (See 
   Lemma \ref{lemma:step-semantics-history-DFA-single-exp} in
  the appendix for details.)
The situation is similar for $\publicmodels$
(where $\Obsset_{\chi} = P$ for all $\chi$) and
for $\decrmodels{}$ (where $\Obsset_{\chi_i} \subseteq \Obsset_{\chi}$
for all standpoint subformulas $\chi_i$ of $\chi$).
In the case of $\publicmodels$ there are history-DFA
where the state space
is contained in $(\prod_{b\in \Ag(\chi)} 2^{S_b}) \times 2^P$.
For $\decrmodels{}$, we assign history-DFA $\cD_v$ to the nodes
in the syntax tree of $v$ that have the shape
$\prod_{w} \pow(\cT_{a_w},\Obsset^{\tinydecr}_w)$
where $w$ ranges 
over all nodes in the syntax subtree of $v$ 
 such that the formula given by $w$ is a standpoint formula
$\StMod{a_w}{\varphi_w}$.
  (See 
   Lemma \ref{lemma:public-semantics-history-DFA-single-exp} 
   and Lemma \ref{lemma:decr-semantics-history-DFA-single-exp}
   in the appendix.)
For $\probsmodels$ under the additional assumption
that $P_a \cap P_b =\varnothing$ for $a,b\in \Ag$ with $a \not= b$,
one can simplify the transition rules for $\cT_{\chi}$
to obtain history-DFA $\cD_{\chi}$ for 
standpoint formulas $\chi=\StMod{a}{\varphi}$ 
where the state space is contained in 
$2^{S_a} \times \prod_{b \in \Ag(\varphi)} 2^{S_b}$.
  (See Lemma \ref{lemma:pobs-disjoint-semantics-history-DFA-single-exp}
  in the appendix.)

%%%%%%%%%%%%%%%%%%%%%%%%%%%%%%%%%%%%%%%%%%%%%

With these observations, we obtain (where we assume 
the above mentioned simplification of the default-history DFAs):

\begin{lemma}
 \label{lemma:single-exp-run-time-step-public-decr}
   For $*\in \{\step,\public,\decr\}$,
   the algorithm of Section \ref{sec:model-checking}
   for $\StpLTL^*$ model checking
   runs in (single) exponential time.
   The same holds for $*=\pobs$ 
   under the additional assumption that the $P_a$'s are pairwise disjoint.
\end{lemma}

%%%%%%%%%%%%%%%%%%%%%%%%%%%%%%%%%%%%%%%%%%%%%%%%%%%%%%%%%%%

Let us now look at $\incrmodels{}$. 
Given a node $v$ in the syntax tree of $\phi$, let $\phi_v$ denote the subformula
of $\phi$ that is represented by $v$, let $\pi_v$ denote
the unique path $\pi_v=v_0 v_1 \ldots v_r$ in the syntax tree from
the root $v_0$ to $v_r = v$, and 
let $\Ag_v$ denote the set of agents
$b\in \Ag$ such that $\phi_w$ has the form $\StMod{b}{\psi_w}$ for some
$w \in \{v_1,\ldots, v_{r-1}\}$.
Then, $Q_{v}^{\tinyincr} = \bigcup_{b \in \Ag_v} P_b$
and $\Obsset_{v}^{\tinyincr} = Q_v^{\tinyincr} \cup P_a$ if
$\phi_v = \StMod{a}{\varphi}$ (see Remark \ref{remark:context-dependency-syntax-tree}).
Let 
$\pi_v'= w_1\ldots w_{\ell}$ be the sequence 
resulting from $\pi_v$
when all nodes $w$ where $\phi_w$ is not a standpoint subformula
are removed.
The observation sets along $\pi_v'$ are increasing, i.e.,
$\varnothing = \Obsset_{w_1}^{\tinyincr} \subseteq \Obsset_{w_2}^{\tinyincr}
    \subseteq \ldots \subseteq \Obsset_{w_{\ell}}^{\tinyincr} \subseteq P$.
Thus, the number of indices $j \in \{2,\ldots,\ell\}$ where
$\Obsset_{w_{j-1}}^{\tinyincr}$ and $\Obsset_{w_j}^{\tinyincr}$ 
are different is bounded
by $N{-}1$ where $N = |\Ag|$. If
$\Obsset_{w_{j-1}}^{\tinyincr} = \Obsset_{w_j}^{\tinyincr}$, 
the history-DFA constructed for $\phi_{w_j}$ is
$\Obsset_{w_{j-1}}^{\tinyincr}$-deterministic.
We can now again apply Remark \ref{remark:exploit-obs-determinism} 
to $\chi=\phi_{w_{j-1}}$ and $\chi_i = \phi_{w_j}$. This yields:

\begin{lemma}
 \label{lemma:N-EXP-incremental-semantics}
   With $N=|\Ag|$,
   our model checking algorithm for $\StpLTL^{\tinyincr}$ 
   is $N$-fold exponentially time bounded.
\end{lemma}

%%%%%%%%%%%%%%%%%%%%%%%%%%%%%%%%%%%%%%%%%%%%%%%%%%%%%%%%%%%%%%%%%%

\subsubsection*{Space bounds}

Lemma \ref{lemma:embdding-pobs-in-LTLK} 
shows that the $\SLTL{\pobs}{d}$ model checking problem 
is polynomially reducible to the LTLK$_{d}$ model checking
problem under the perfect-recall semantics.
The latter is known to be $(d{-}1)$-EXPSPACE-complete 
for $d \geqslant 2$ and PSPACE-complete for $d=1$
\cite{MC-LTLK-and-beyond-2024}.
In combination with the known
PSPACE lower bound for the LTL model checking problem
\cite{SistlaClarke-JACM85} and 
Lemmas \ref{lemma:pobs=decr=incr-for-alternation-depth-1}
and \ref{lemma:embdding-pobs-in-LTLK} 
we obtain:

\begin{corollary} 
 \label{cor:complexity-pobs}
    For $*\in \{\pobs,\decr,\incr\}$,
    the $\StpLTL^*_1$ model checking problem is PSPACE-complete.
    For $d \geqslant 2$,
    the $\SLTL{$\pobs$}{d}$ model checking problem 
    is in $(d{-}1)$-EXPSPACE.
\end{corollary}

For the step and public-history semantics, the
reduction provided in the proof of Lemma \ref{lemma:embdding-pobs-in-LTLK}
can be adapted by dealing with an LTLK structure over a single agent.
Thus, the generated LTLK formula has alternation depth at most 1.
See Lemma \ref{lemma:step-public-in-LTLKone} in the appendix.

\begin{theorem}
 \label{thm:step-public-PSPACE-completeness}
   For $*\in \{\step,\public\}$,
   the $\StpLTL^*$ model checking problem 
   is polynomially reducible to the LTLK$_1$ model checking
   problem, and therefore PSPACE-complete.
\end{theorem}

Lemma \ref{lemma:embdding-pobs-in-LTLK} and 
Theorem \ref{thm:step-public-PSPACE-completeness} show that the
$\StpLTL^*$ model checking problem for $*\in \{\step,\public,\pobs\}$
can be viewed as an instance of the LTLK model checking problem.
An analogous statement holds for $\decrmodels{}$ and $\incrmodels{}$.
The idea is to modify the reduction described in the
proof of Lemma \ref{lemma:embdding-pobs-in-LTLK} by adding new
agents 
for all standpoint subformulas $\chi = \StMod{a}{\psi}$ where
$P_a \not= \Obsset_{\chi}$.
This yields a polynomial reduction of the 
$\SLTL{*}{d}$ model checking problem 
to the LTLK$_M$ model checking
problem where $M = \min \{|\Ag|,d\}$ and $*\in \{\incr, \decr\}$.
See Lemma \ref{lemma:embedding-incr-in-LTLK} in the
appendix.
With the above mentioned results of \cite{MC-LTLK-and-beyond-2024}, we can improve
the complexity-theoretic upper bound stated in
Lemma \ref{lemma:N-EXP-incremental-semantics}:

\begin{corollary}
\label{cor:incr-N-1-EXPSPACE}  
  The $\StpLTL^{\tinyincr}$ model checking problem is in
  $(N{-}1)$-EXPSPACE where $N = |\Ag|$.
\end{corollary}

%% file: Sections/conclusion.tex
\section{Conclusion}

We considered five different semantics for $\StpLTL$ that differ in the amount of information the agents can extract from the history, and presented a generic model-checking algorithm applicable to all five semantics.
The computational complexity of the algorithm, however, varies between the semantics due to the different numbers of necessary applications of the powerset construction.
More precisely,
the generic $\StpLTL^*$ model checking algorithm 
is $m$-fold exponentially time-bounded
with $m=1$ for $* \in \{\step,\public,\decr\}$, $m = |\Ag|$ for $* = \incr$,
and $m = \ad(\phi)$ for $*=\pobs$. 

Algorithms with improved space complexity for $\stepmodels$, $\publicmodels$,
$\pobsmodels$, and $\incrmodels{}$ are obtained via embeddings into LTLK.
To match the space bounds,
the model-checking algorithm of Section 5 can be adapted 
by combining classical on-the-fly
automata-based LTL model checking techniques with an on-the-fly
construction of history-DFA, 
similar to the
techniques proposed in \cite{MC-LTLK-and-beyond-2024} for CTL*K (for more details, see Appendix \ref{sec:m-minus-1-EXPSPACE}).
Analyzing if analogous techniques are applicable to obtain
a polynomially space-bounded algorithm for $\decrmodels{}$, 
providing lower bounds beyond PSPACE for $\incrmodels{}$ and $\probsmodels$, as well as an experimental evaluation of the presented algorithm remain as future work.

%% file: Sections/appendix.tex
%%%%%%%%%%%%%%%%%%%%%%%%%%%%%%%%%%%%%%%%%%%%%%%%%%%%%%%%%%%%%%%%%%%

\begin{appendix}

\section*{Appendix}

%%%% 
\input{Sections/app-LTL.tex}

%%%% 
\input{Sections/app-logic.tex}

%%%% 
\input{Sections/app-mc.tex}

%%%% 

\input{Sections/app-complexity.tex}

\end{appendix}

%% file: Sections/app-LTL.tex
\section{Trace-position versus future-history pairs}

The standard satisfaction relation $\LTLmodels$ of LTL interpretes formulas over
traces. An alternative way often used for extensions of LTL (such as standpoint LTL \cite{SLTL-KR23,Complexity-SLTL-ECAI24}, LTL with past modalities or LTLK \cite{MC-LTLK-and-beyond-2024}) considers trace-position pairs $(\rho,n)\in (2^P)^{\omega}\times \Nat$ and defines the satisfaction relation via
\begin{center}
\begin{tabular}{lcl}
   $(\rho,n) \models true$ \\
   $(\rho,n) \models p$ & iff & $p\in \rho[n]$ \\
   $(\rho,n) \models \varphi_1 \wedge \varphi_2$ & 
    iff & $(\rho,n) \models \varphi_1$ and $(\rho,n) \models \varphi_2$ 
   \\
   $(\rho,n) \models \neg \varphi$ & 
    iff & $(\rho,n) \not\models \varphi$
   \\
   $(\rho,n) \models \neXt \varphi$ & 
    iff & $(\rho,n{+}1) \models \varphi$
   \\
   $(\rho,n) \models \varphi_1 \Until \varphi_2$ & 
    iff & there exists $\ell \in \Nat$ with $\ell \geqslant n$,
   \\
   & & $(\rho,k) \models \varphi_1$ for all 
   $k \in \{n,\ldots,\ell{-}1\}$
   \\
   & & and $(\rho,\ell) \models \varphi_2$ 
\end{tabular}
\end{center}

We use here a semantics over future-history pairs 
$(f,h)\in (2^P)^{\omega}\times (2^P)^+$ with
$\last(h)=\first(h)$.

Obviously, there is a one-to-one correspondence between such 
future-history pairs $(f,h)$ and trace-position pairs $(\rho,n)$,
namely, $(f,h) \mapsto (h \circ f, |h|{-}1)$ and
$(\rho,n) \mapsto (\suffix{\rho}{n},\prefix{\rho}{n})$
where $h \circ f$ arises from the concatenation of $h$ and $f$ by removing the
duplicated symbol $\last(h)=\first(h)$.

Moreover, the semantics of LTL given by the rules
shown in Figure \ref{fig:semantics} and the satisfaction relation over
trace-position pairs is consistent with these 
translations of future-history pairs into
trace-position pairs and vice versa.

%% file: Sections/app-logic.tex
\section{Further details and proofs for Section \ref{sec:logic}}

\label{sec:app-logic}

%%%%%%%%%%

\subsection{Alternation depth}

The \emph{alternation depth} $\ad(\varphi)$ of $\StpLTL$ formula $\varphi$ 
is the maximal number of alternations between standpoint modalities
for different agents in $\varphi$. 
For instance, 
$\ad(\StMod{a}{(p \wedge \neXt \DualStMod{b}{q})})=2$, while
$\ad(\StMod{a}{(p \wedge \neXt \DualStMod{a}{q})})= 
 \ad(\StMod{a}{p} \wedge \neXt \DualStMod{b}{q}) = 1$ 
if $a \not= b$.

\begin{definition}[Alternation depth of $\StpLTL$ formulas]
\label{def:alternation-depth}
{\rm
The alternation depth of $\StpLTL$ formulas is defined inductively:
\begin{itemize}
\item
   $\ad(\true)=\ad(p)=0$ for $p\in P$,
\item
   $\ad(\varphi_1 \wedge \varphi_2)=\ad(\varphi_1 \Until \varphi_2)=
      \max \{\ad(\varphi_1),\ad(\varphi_2)\}$,
\item
   $\ad(\neg \varphi)=\ad(\neXt \varphi)=\ad(\varphi)$.
\end{itemize}
For standpoint formulas $\phi = \StMod{a}{\varphi}$, 
the definition of $\ad(\phi)$ is as follows.
If $\varphi$ is an LTL formula then
$\ad(\StMod{a}{\varphi})= 1$.
Otherwise let $\chi_1=\StMod{b_1}{\psi}, \ldots, \chi_k=\StMod{b_k}{\psi_k}$ 
be the maximal standpoint subformulas of $\varphi$.
We may suppose an enumeration such that $a = b_1 = \ldots = b_{\ell}$ and
$a \notin \{b_{\ell +1},\ldots,b_k\}$.
Then, 
$\ad(\phi)$ is the maximum of
$\max \{ \ad(\chi_i): i =1,\ldots,\ell \}$ and
$\max \{ \ad(\chi_i): i =\ell{+}1,\ldots, k \} + 1$.
  }
\end{definition}

%%%%%%%%%

\subsection{Context-dependency of the decremental and incremental semantics}

   Figure \ref{fig:syntax-tree-incd-decr-semantics}  
   serves to 
   illustrate the definition of the decremental and incremental semantics
   using the sets
   $Q^{\incr}_v$ and  $Q^{\decr}_v$ for the nodes in
   the syntax tree of a given formula $\phi$ defined in
   Remark \ref{remark:context-dependency-syntax-tree}.

\begin{figure*}[ht]
	\begin{center}
		\begin{tikzpicture}[->,>=stealth',shorten >=1pt,auto, semithick]
			\tikzstyle{every state} = [text = black]
			
			\node[state,scale=0.75] (root) [fill = cyan!50, inner sep = 0pt] {$\lor$};
			\node[,scale=0.75] (temp) [below of = root, node distance = 1cm] {};
			\node[state,scale=0.75] (a) [left of = temp, node distance = 1.5cm, fill = cyan!50, inner sep = 0pt] {$\StMod{a}{}$}; 
			\node[state,scale=0.75] (b) [right of = temp, node distance = 1.5cm,fill=cyan!50, inner sep = 0pt] {$\StMod{b}{}$}; 
			\node[state,scale=0.75] (p1) [below of = a, node distance = 1.5cm, fill = magenta!50, inner sep = 0pt] {$p$}; 
			\node[state,scale=0.75] (and) [below of = b, node distance = 1.5cm, fill = green, inner sep = 0pt] {$\land$}; 
			\node[,scale=0.75] (temp2) [below of = and, node distance = 1cm, inner sep = 0pt] {};
			\node[state,scale=0.75] (q) [left of = temp2, node distance = 1cm, fill = green, inner sep = 0pt] {$q$}; 
			\node[state,scale=0.75] (next) [right of = temp2, node distance = 1cm, fill = green, inner sep = 0pt] {$\neXt$}; 
			\node[state,scale=0.75] (a2) [below of = next, node distance = 1.5cm , fill =green, inner sep = 0pt] {$\StMod{a}{}$}; 
			\node[state,scale=0.75] (p2) [below of = a2, node distance = 1.5cm, fill = yellow, inner sep = 0pt] {$p$}; 
			
			\path 
			(root) edge (a)
			(root) edge (b)
			(a) edge (p1)
			(b) edge (and)
			(and) edge (q)
			(and) edge (next)
			(next) edge (a2)
			(a2) edge (p2)
			;
			
			\node [right of = root, node distance = 0.75cm,scale=0.8] {$\decrmodels{P}$}; 
			\node [left of = a, node distance = 0.75cm,scale=0.8] {$\decrmodels{P}$}; 
			\node [left of = p1, node distance = 0.75cm,scale=0.8] {$\decrmodels{P_a}$}; 
			\node [right of = b, node distance = 0.75cm,scale=0.8] {$\decrmodels{P}$}; 
			\node [right of = and, node distance = 0.75cm,scale=0.8] {$\decrmodels{P_b}$}; 
			\node [left of = q, node distance = 0.75cm,scale=0.8] {$\decrmodels{P_b}$}; 
			\node [right of = next, node distance = 0.75cm,scale=0.8] {$\decrmodels{P_b}$}; 
			\node [right of = a2, node distance = 0.75cm,scale=0.8] {$\decrmodels{P_b}$}; 
			\node [right of = p2, node distance = 0.82cm,scale=0.8] {$\decrmodels{P_b \cap P_a}$}; 
			
			%%rhs
			
			\node[state,scale=0.75] (rootr) [fill = cyan!50, inner sep = 0pt, right of = root, node distance = 8cm] {$\lor$};
			\node[,scale=0.75] (tempr) [below of = rootr, node distance = 1cm] {};
			\node[state,scale=0.75] (ar) [left of = tempr, node distance = 1.5cm, fill = cyan!50, inner sep = 0pt] {$\StMod{a}{}$}; 
			\node[state,scale=0.75] (br) [right of = tempr, node distance = 1.5cm,fill=cyan!50, inner sep = 0pt] {$\StMod{b}{}$}; 
			\node[state,scale=0.75] (p1r) [below of = ar, node distance = 1.5cm, fill = magenta!50, inner sep = 0pt] {$p$}; 
			\node[state,scale=0.75] (andr) [below of = br, node distance = 1.5cm, fill = green, inner sep = 0pt] {$\land$}; 
			\node [,scale=0.75] (temp2r) [below of = andr, node distance = 1cm, inner sep = 0pt] {};
			\node[state,scale=0.75] (qr) [left of = temp2r, node distance = 1cm, fill = green, inner sep = 0pt] {$q$}; 
			\node[state,scale=0.75] (nextr) [right of = temp2r, node distance = 1cm, fill = green, inner sep = 0pt] {$\neXt$}; 
			\node[state,scale=0.75] (a2r) [below of = nextr, node distance = 1.5cm , fill =green, inner sep = 0pt] {$\StMod{a}{}$}; 
			\node[state,scale=0.75] (p2r) [below of = a2r, node distance = 1.5cm, fill = yellow, inner sep = 0pt] {$p$}; 
			
			\path 
			(rootr) edge (ar)
			(rootr) edge (br)
			(ar) edge (p1r)
			(br) edge (andr)
			(andr) edge (qr)
			(andr) edge (nextr)
			(nextr) edge (a2r)
			(a2r) edge (p2r)
			;
			
			\node [right of = rootr, node distance = 0.75cm,scale=0.75] {$\incrmodels{\varnothing}$}; 
			\node [left of = ar, node distance = 0.75cm,scale=0.75] {$\incrmodels{\varnothing}$}; 
			\node [left of = p1r, node distance = 0.75cm,scale=0.75] {$\incrmodels{P_a}$}; 
			\node [right of = br, node distance = 0.75cm,scale=0.75] {$\incrmodels{\varnothing}$}; 
			\node [right of = andr, node distance = .75cm,scale=0.75] {$\incrmodels{P_b}$}; 
			\node [left of = qr, node distance = 0.75cm,scale=0.75] {$\incrmodels{P_b}$}; 
			\node [right of = nextr, node distance = 0.75cm,scale=0.75] {$\incrmodels{P_b}$}; 
			\node [right of = a2r, node distance = 0.75cm,scale=0.75] {$\incrmodels{P_b}$}; 
			\node [right of = p2r, node distance = 0.82cm,scale=0.75] {$\incrmodels{P_b \cup P_a}$}; 
		\end{tikzpicture}
		\caption{Syntax tree for
			$\phi = \StMod{a}{p} \vee 
			\StMod{b}{(q \wedge \neXt \StMod{a}{p})}$
			and the corresponding satisfaction relations
			$\decrmodels{Q}$ and $\incrmodels{Q}$ for the subformulas.}
		\label{fig:syntax-tree-incd-decr-semantics} 
\end{center}
\end{figure*}

%%%%%%%%%%%%%%%%%%%%%%%%%%%%%%%%%%%%%%%%%%%%%%%%%%%%%%%%%%%%%%%%%%%%%%%%%%%

\subsection{Properties of $\StpLTL$}

\begin{lemma}[(Cf.~Lemma \ref{obs-a-lemma})]
 \label{app:obs-a-lemma}
    If $h_1, h_2 \in (2^P)^+$ with $\obs_a(h_1)=\obs_a(h_2)$  
    then
    $(*,h_1) \models \StMod{a}{\varphi}$ iff 
    $(*,h_2) \models \StMod{a}{\varphi}$.
\end{lemma}

\begin{proof}
By symmetry it suffices to show that
$(*,h_1) \models \StMod{a}{\varphi}$ implies
$(*,h_2) \models \StMod{a}{\varphi}$.
Suppose $(*,h_1) \models \StMod{a}{\varphi}$. 
Hence, there exist $h' \in (2^{P})^+$,
a state $t\in \Reach(\cT_a,h')$ 
and a trace
$f'\in \Traces^P(\cT_a,t)$ such that
$\last(h')=\first(f')$,
$\obs_a(h_1)=\obs_a(h')$ and $(f',h') \models \varphi$.
But then, $\obs_a(h_2)= \obs_a(h_1)=\obs_a(h')$
and $(*,h_2) \models \StMod{a}{\varphi}$.
\end{proof}

%%%%%%%%%%%%%%%%%%%%%%%%%%%%%%%%%%%%%%%%%%%%%%%%%%%%%%%%%%%%%%%%%%%%%

\begin{lemma}[Cf.~Lemma \ref{lemma:pobs=decr=incr-for-alternation-depth-1}]
\label{app:lemma:pobs=decr=incr-for-alternation-depth-1}
If $\phi$ is a $\StpLTL$ formula of alternation depth 1 then
$\fT \pobsmodels \phi$ 
iff $\fT \decrmodels{} \phi$ 
iff $\fT \incrmodels{} \phi$, while
$\fT \not\pobsmodels \phi$ 
and $\fT \models \phi$ 
(or vice versa)
is possible where ${\models} \in \{\stepmodels,\publicmodels\}$.
Likewise,
$\fT \not\publicmodels \phi$ 
and $\fT \stepmodels \phi$ (or vice versa) is possible.
\end{lemma}

\begin{proof}
The first part follows from the observation that
for every occurrence of a subformula of $\phi$ 
that is in the scope of a standpoint modality
$\StMod{a}{}$, the observation set used to extract information
from the history is $P_a$ under the pure observation-based,
the decremental and the incremental semantics.

To distinguish the pure observation-based from the public semantics,
consider $P=\{p\}$, $P_a=\varnothing$, $\phi = \StMod{a}{p}$ and suppose 
$p \notin L_0(\init_0)$. Then, $\fT \pobsmodels \phi$ 
(as $\cT_a$ has no information on $\cT_0$'s initial state and
may guess the satisfaction of $p$ in the initial state)
and $\fT \not\publicmodels \phi$ 
(as $\cT_a$ has observed that $p$ does not hold in $\cT_0$'s initial state).

To distinguish $\pobsmodels$ and $\publicmodels$ from $\stepmodels$,
regard 
$P=P_a=\{p\}$, $\phi = \StMod{a}{p}$ and suppose
that $p \in L_a(\init_a) \setminus L_0(\init_0)$. 
Then, $\fT \not\pobsmodels \phi$ and $\fT \not\publicmodels \phi$,
while $\fT \stepmodels \phi$.
\end{proof}

%%%%%%%%%%%%%%%%%%%%%%%%%%%%%%%%%%%%%%%%%%%%%%%%%%%%%%%%%%%%%%%%%%%%%%%%%%%%
%%%%%
%%%%%    Embedding POBS to LTLK
%%%%%
%%%%%%%%%%%%%%%%%%%%%%%%%%%%%%%%%%%%%%%%%%%%%%%%%%%%%%%%%%%%%%%%%%%%%%%%%%%%

\subsection{Embedding of $\StpLTL^{\tinypobs}$ into LTLK}

\begin{lemma}[Cf.~Lemma \ref{lemma:embdding-pobs-in-LTLK}]
	\label{app:lemma:embdding-pobs-in-LTLK}
	\label{app:lemma:embedding-pobs-in-LTLK}
	Given a pair $(\fT,\phi)$ consisting of a $\StpLTL$ structure 
	$\fT = (\cT_0, (\cT_a)_{a\in \Ag})$ and a $\StpLTL$ formula $\phi$, 
	one can construct in polynomial time
	an LTLK structure $\ltlk{\fT} = (\cT',(\sim_a)_{a\in \Ag})$
	and an LTLK formula $\ltlk{\phi}$ such that:
	\begin{enumerate}
		\item [(1)] 
		$\fT \pobsmodels \phi$ if and only if 
		$\ltlk{\fT} \LTLKmodels \ltlk{\phi}$ 
		\item [(2)] 
		$\phi$ and $\ltlk{\phi}$ have the same alternation depth.
	\end{enumerate}
\end{lemma}

\begin{proof}
	We may assume w.l.o.g. that the labelings of the initial states
	in $\cT_0$ and the $\cT_a$'s are consistent in the sense that
	$L_0(\init_0) \cap P_a = L_a(\init_a)$ for each $a \in \Ag$.
	If this is not the case, we extend $\cT_a$ for $a\in \Ag \cup \{0\}$ by
	a fresh initial state $\init_a'$ with $L_a(\init_a') = \varnothing$ and
	transitions $\init_a' \to_a \init_a$ and replace the given $\StpLTL$ formula
	$\phi$ with $\neXt \phi$.
	
	For each $a \in \Ag \cup \{0\}$, 
	we first switch from $\cT_a$ to its completion $\cT_a^{\bot}$ which extends
	$\cT_a$ by states $\bot_a^Q$ for $Q \subseteq P_a$.
	The state space of $\cT_a^{\bot}$ is 
	$S_a^{\bot} = S_a \cup \{\bot_a^Q : Q \subseteq P_a\}$, and the initial state is $\init_a$.
	The transition relation of $\cT_a^{\bot}$ extends $\to_a$ by:
	\begin{itemize}
		\item
		$s_a \to_a \bot_a^Q$ if there is no state $s'$ in $\cT_a$ with
		$L_a(s')=Q$ and $s \to_a s'$ 
		\item
		$\bot_a^R \to_a \bot_a^Q$ for all $R, Q \subseteq P_a$.
	\end{itemize}
	The set of atomic propositions in $\cT_a^{\bot}$ is 
	$P_a^{\bot}=P_a \cup \{\bot_a\}$.
	The labeling of the original states is unchanged, i.e.,
	$L_a^{\bot}(s)=L_a(s)$, while
	for the fresh states $L_a^{\bot}(\bot_a^Q)=Q \cup \{\bot_a\}$.
	
	Thus, for each word $\rho = L_a(\init_a)\rho'$, where  $\rho'\in (2^{P_a})^{\omega}$ there exists a trace 
	$f \in \Traces(\cT^{\bot}_a)$
	with $\proj{f}{P_a}=\rho$. 
	(In this sense, 
	$\cT_a^{\bot}$ is the completion of $\cT_a$.) 
	Moreover, all $\rho \in \Traces(\cT_a^{\bot})$ satisfy the LTL formula
	$\Box \neg \bot_a \vee \Diamond \Box \bot_a$.
	The original traces of $\cT_a$ can be recovered from $\cT_a^{\bot}$
	as we have:
	\begin{center}
		$\Traces(\cT_a,s) \ = \ 
		\bigl\{ \proj{\rho}{P_a} :  
		\rho \in \Traces(\cT_a^{\bot},s), \ \rho \LTLmodels \Box \neg \bot_a  
		\bigr\}$
	\end{center}
	
	The transition system $\cT' = (S',\to',\init',P',L')$ 
	of the constructed LTLK structure $\ltlk{\fT}$
	arises by the synchronous product of
	the transition systems $\cT_a^{\bot}$. 
	That is, the state space $S'$ of $\cT'$ consists of all
	tuples $(s,(s_a)_{s\in \Ag}) \in S_0^{\bot} \times \prod_{a\in \Ag} S_a^{\bot}$ 
	such that $L_0^{\bot}(s) \cap P_a = L_a^{\bot}(s_a)\cap P_a$ for each $a \in \Ag$.
	The set $P'$ of atomic propositions is 
	$P'=P \cup \{\bot_a : a \in \Ag \cup \{0\}\}$.
	The labeling function is given by 
	$L'(s, (s_a)_{a\in \Ag}) = L_0(s) \cup 
	\{\bot_a : a \in \Ag \cup \{0\}, \bot_a \in L_a^{\bot}(s_a)  \}$.
	The initial state of $\cT'$ is 
	$\init' = (\init_0, (\init_a)_{a\in \Ag})$. By the assumption that
	the labelings of the initial states in $\fT$ are consistent, $\init'$ is indeed
	an element of $S'$.
	The transition relation $\to'$ in $\cT'$ is defined as follows. If
	$\sigma = (s,(s_a)_{a\in \Ag})$ and $\theta =(t,(t_a)_{a\in \Ag})$ are
	elements of $S'$ then
	$\sigma \to' \theta$
	if and only if $s \to_0 t$ in $\cT_0^{\bot}$ and $s_a \to_a t_a$ in 
	$\cT_a^{\bot}$.
	For $a \in \Ag$, the equivalence relation $\sim_a$ on $S'$ is given by:
	$\sigma \sim_a \theta$ if and only if 
	$L'(\sigma) \cap P_a^{\bot} = L'(\theta) \cap P_a^{\bot}$.
	By construction, for each trace $\rho \in (2^P)^{\omega}$ there
	is at least one path $\pi$ in $\cT'$ with $\proj{\trace(\pi)}{P}=\rho$.

	The essential idea for
	translating the given $\StpLTL$ formula $\phi$ into an ``equivalent''
	LTLK formula $\ltlk{\phi}$ is to identify 
	$\DualStMod{a}{\varphi}$ with 
	$K_a (\Box \neg \bot_a \, \to \, \varphi)$.
	The precise definition of $\ltlk{\phi}$ is as follows.
	We first provide a translation $\varphi \leadsto \primed{\varphi}$
	of $\StpLTL$ formulas $\varphi$ into LTLK formulas $\primed{\varphi}$
	by structural induction:
	If $\varphi \in \{\true\} \cup P$ then $\primed{\varphi}=\varphi$.
	In the step of induction, we put
	$\primed{(\neg \varphi)}=\neg \primed{\varphi}$, 
	$\primed{(\varphi \wedge \psi)} = \primed{\varphi}\wedge \primed{\psi}$,
	$\primed{(\neXt \varphi)}=\neXt \primed{\varphi}$,
	$\primed{(\varphi \Until \psi)} = \primed{\varphi} \Until \primed{\psi}$ and
	$\primed{(\StMod{a}{\varphi})} = 
	\overline{K}_a (\primed{\varphi} \wedge \Box \neg \bot_a)$
	where $\overline{K}_a$ is the dual of $K_a$ given by
	$\overline{K}_a \psi = \neg K_a \neg \psi$.
	Finally, we define $\ltlk{\phi}= \Box \neg \bot_0 \to \primed{\phi}$.
	
	Obviously, the translation of $(\fT,\phi)$ 
	into $(\ltlk{\fT},\ltlk{\phi})$ is polynomial.
	It remains to show that $\fT \pobsmodels \phi$ 
	iff $\ltlk{\fT} \LTLKmodels \ltlk{\phi}$. 
	By structural induction we obtain that for each path $\pi$ in $\cT'$
	and each $n \in \Nat$ we have:
	\begin{center}
		$(\pi,n) \LTLKmodels \primed{\varphi}$
		\ iff \ 
		$(\proj{\trace(\suffix{\pi}{n})}{P},
		\proj{\trace(\prefix{\pi}{n})}{P}) \pobsmodels \varphi$
	\end{center}
	where $\LTLKmodels$ refers to the LTLK semantics over the
	structure $\ltlk{\fT}$ and
	$\pobsmodels$ to the pure observation-based semantics of $\StpLTL$
	over the original $\StpLTL$ structure $\fT$.
	Then, $\ltlk{\fT} \LTLKmodels \ltlk{\phi}$ iff 
	$(\pi,0) \LTLKmodels \Box \neg \bot_0 \to \primed{\phi}$
	for every path $\pi \in \Paths(\cT')$
	iff
	$(\pi,0) \LTLKmodels \primed{\phi}$
	for every path $\pi \in \Paths(\cT')$ 
	with $\trace(\pi) \LTLmodels \Box \neg \bot_0$
	iff
	$(f,\first(f)) \pobsmodels \phi$ for every
	trace $f \in \Traces(\cT_0)$
	iff $\fT \pobsmodels \phi$.
	Here, we use the fact that $\Traces(\cT_0)$ equals 
	$\{\proj{\trace(\pi)}{P} : 
	\pi \in \Paths(\cT'), \trace(\pi) \LTLmodels \Box \neg \bot_0 \}$.
\end{proof}

%% file: Sections/app-mc.tex
%%%%%%%%%%%%%%%%%%%%%%%%%%%%%%%%%%%%%%%%%%%%%%%%%%%%%%%%%%%%%%%%%%%
%%%%%%%%%%%%%%%%%%%%%%%%%%%%%%%%%%%%%%%%%%%%%%%%%%%%%%%%%%%%%%%%%%%
%%%%%%
%%%%%%     Details and proofs for the model checking section
%%%%%%
%%%%%%%%%%%%%%%%%%%%%%%%%%%%%%%%%%%%%%%%%%%%%%%%%%%%%%%%%%%%%%%%%%%
%%%%%%%%%%%%%%%%%%%%%%%%%%%%%%%%%%%%%%%%%%%%%%%%%%%%%%%%%%%%%%%%%%%

\section{Further details and proofs for Section \ref{sec:model-checking}}

\label{sec:app-mc}

\subsection{Transition system $\cT_a^R$}

As explained in Section \ref{sec:model-checking}, the treatment
of the incremental and public-history semantics requires the
switch from the transition systems $\cT_a$ over $P_a$ to transition systems
$\cT_a^R$ for some superset $R$ of $P_a$.
The behavior of $\cT_a^R$ is the same as $\cT_a$, except that $\cT_a^R$
additionally makes nondeterministic guesses for the truth values
of the atomic propositions in $R \setminus P_a$.

\begin{definition}
\label{def:T-a-R}
{\rm
Let  $a \in \Ag$ and $R \in 2^P$ such that $P_a \subseteq R$.
Then,
$\cT_a^R = (S_a^R,\to_a^R,\Init_a^R,R,L_a^{R})$ is defined as follows.
\begin{itemize}
\item
  The state space $S_a^R$  
consists of the pairs $(s,O)$
with $s \in S_a$ and $O \subseteq R$ such that
$L_a(s) = O \cap P_a$.  
\item
  The set of initial states is $\Init_a^R = \{ (\init_a, O) :
   L_a(\init_a) = O \cap P_a \}$.  
\item
 The transition relation is defined as follows.
  If $(s,O), (s',O') \in S_a^R$ then
$(s,O) \to_a^R (s',O')$ if and only if $s \to_a s'$.
\item
  The labeling function is given by
$L_a^{R}(s,O)= O$.
\end{itemize}
  }
\end{definition}

Obviously, we then have
$\Traces^P(\cT_a^{R},(s,O)) \ = \ 
  \bigl\{ f \in \Traces^P(\cT_a,s) : \first(f) \cap R =O \bigr\}
$
for each state $s\in S_a$ and each $O \subseteq R$ with
$O \cap P_a=L_a(s)$. 
Furthermore, for each history $h \in (2^P)^+$:
$\Reach(\cT_a^R,h) \ = \ 
  \bigl\{ (s,O) : s \in \Reach(\cT_a,h), O = \last(h) \cap R \bigr\}$.

%%%%%%%%%%%%%%%%%%%%%%%%%%%%%%%%%%%%%%%%%%%%%%%%%%%%%%%%%%%%%%%%%%%%%%%%%%%%
%%%%%
%%%%%    LTL model checking
%%%%%
%%%%%%%%%%%%%%%%%%%%%%%%%%%%%%%%%%%%%%%%%%%%%%%%%%%%%%%%%%%%%%%%%%%%%%%%%%%%

\subsection{LTL model checking subroutines}

\label{sec:LTL-MC}

For both the basis and the step of induction, we rely on (a mild variant of)
standard
LTL model checking techniques \cite{VarWolp-LICS86,VarWolp-InfComp94}
to evaluate the arguments $\varphi$ 
of standpoint subformulas $\StMod{a}{\varphi}$.
We state this in a uniform manner in Lemma \ref{lemma:MC-LTL-subformulas}
for arbitrary transition systems $\cT$ over a set of atomic propositions
$R$ and LTL formulas $\varphi$ over $\cP$ where $R$ is a subset of $\cP$.
Then, $(\cT,R,\cP)= (\cT_a^{R},R,P)$ in the basis of induction and
$(\cT,R,\cP) = (\cT_{\chi},P_{\chi},P\cup \{p_1,\ldots,p_k\})$ 
in the step of induction where the $p_i$'s are the fresh atomic
propositions that are introduced for the maximal standpoint
subformulas of $\varphi$.

\begin{lemma}
 \label{lemma:MC-LTL-subformulas}
Let $\cP$ be a set of atomic propositions and $R \subseteq \cP$.
Given a transition system $\cT = (S,\to,\Init,R,L)$ and
an LTL formula $\varphi$ over $\cP$,
the set 
\begin{center}
  $\Sat_{\cT}(\exists \varphi)
  \ = \ 
  \bigl\{ s \in S : 
           \exists f \in \Traces^{\cP}(\cT,s) \text{ s.t. } f\LTLmodels \varphi
  \bigr\}$
\end{center}
is computable in time exponential in the
length of $\varphi$ and polynomial in the size of $\cT$. 
Furthermore, the problem to decide whether
$\Init \cap \Sat_{\cT}(\exists \varphi) \not= \varnothing$ is in $\PSPACE$.
\end{lemma}

\begin{proof}
We apply standard techniques to 
build a nondeterministic B\"uchi automaton (NBA) 
$\cA = \cA_{\psi} = (Y,\delta_{\cA},\init_{\cA},F)$ 
for $\psi$ over the alphabet $2^{\cP}$. That is, the accepted language
$\cL(\cA)$ of $\cA$ are the infinite words $f \in 2^{\cP}$ such that $f \LTLmodels \psi$.
We now consider the product 
$\cT \bowtie \cA = (S',\to'\nolinebreak,\Init',Y,L')$ 
defined as follows:
\begin{itemize}
\item
  The state space $S'$ is $S \times Y$.
\item
  The transition relation ${\to}' \, \subseteq \, (S \times Y)^2$ is defined 
  as follows.
  If $(s,y), (s',y')\in S \times Y$ then
  $(s,y) \to' (s',y')$ iff 
        there exists $H \subseteq \cP$
         such that 
         $s \to s'$, $H \cap R=L(s)$ and $y'\in \delta_{\cA}(y,H)$.
\item
  The labels of the states are mostly irrelevant, we only 
  need to distinguish states
  $(s,y)$ where $y \in F$ from those where $y \notin F$.
  Formally, we can deal with the state space of $\cA$ as atomic propositions
  and
  the labeling function $L' : S' \to Y$ given by
  $L'(s,y) = \{y\}$.
\item
  The set of initial states is
  $\Init' = \bigcup_{s\in S} \Init_{\cT \bowtie \cA}(s)$
  where for $s \in S$,
  $
   \Init_{\cT\bowtie \cA}(s) \ = \ 
   \bigl\{ \ (s,y) \ : \ \exists H \subseteq \cP \text{ s.t. }
             H \cap R = L(s), \
             y \in \delta_{\cA}(\init_{\cA},H) \ 
   \bigr\}$.
\end{itemize}
We then apply standard graph algorithms (backward search from the SCCs)
to determine the set 
$\Sat_{\cT \bowtie \cA}(\exists \Box \Diamond F)$ 
consisting of all states
$(s,y)$ in $\Init'$
that have a path visiting
$S \times F$ infinitely often.
Then, $\Sat_{\cT}(\exists \varphi)$ is the set
of all states $s$ in $\cT$
with
$\Init_{\cT\bowtie \cA}(s) \cap \Sat_{\cT \bowtie \cA}(\exists \Box \Diamond F) \neq \varnothing$.
Thus, $\Sat_{\cT}(\exists \varphi)$ 
is computable with known techniques 
in time exponential in the length of $\varphi$ and polynomial in the
size of $\cT$.
The polynomial space bound for the decision problem 
can be obtained with the same techniques
as for standard automata-based
LTL model checking \cite{VarWolp-LICS86,VarWolp-InfComp94}.
\end{proof}

%%%%%%%%%%%%%%%%%%%%%%%%%%%%%%%%%%%%%%%%%%%%%%%%%%%%%%%%%%%%%%%%%%%%%%%%%%%
%%%%%
%%%%%    soundness of the default history-DFA
%%%%%
%%%%%%%%%%%%%%%%%%%%%%%%%%%%%%%%%%%%%%%%%%%%%%%%%%%%%%%%%%%%%%%%%%%%%%%%%%%

\subsection{Soundness of the default history-DFA}

\label{sec:soundness-history-DFA}

\begin{lemma}
 \label{lemma:DFA-construction}
 \label{lemma:properties-default-history-DFA}
Let $\cD = \pow(\cT,\Obsset,U)$ be as
in Definition \ref{def:default-history-DFA}.
Then: 
\begin{enumerate}
\item [(a)]
  For $h \in (2^P)^+$,
  $\delta_{\cD}(\init_{\cD},h)$ 
  equals the set of all states $s$ in $\cT$ 
  such that
  there exist an initial, finite path $\pi$ in $\cT$
  to $s$ and a word $h' \in (2^P)^+$ with
  $\trace(\pi)=\proj{h'}{R}$ and
  $\proj{h}{\Obsset} = \proj{h'}{\Obsset}$.

\item [(b)]
  $\cD$ is $\Obsset$-deterministic, i.e.,
  if $h_1,h_2 \in (2^P)^+$ with 
  $\proj{h_1}{\Obsset}=\proj{h_2}{\Obsset}$ then
  $\delta_{\cD}(\init_{\cD},h_1)=\delta_{\cD}(\init_{\cD},h_2)$.

\item [(c)]
  The language $\cL(\cD)$ of $\cD$
  equals the set of all words $h \in (2^P)^+$ such that 
  there exist $h'\in (2^P)^+$ 
  with $\proj{h}{\Obsset}=\proj{h'}{\Obsset}$ and
  $U \cap \Reach(\cT,h') \not= \varnothing$.
\end{enumerate}
\end{lemma}

\begin{proof}
Statement (a) can be shown
by induction on the length of $h$.
Statement (b) is a direct consequence of (a).

To prove statement (c),
let $\Hist_{\cT}(U)$ denote the set of all histories $h \in (2^P)^+$ such that 
  there exist $h'\in (2^P)^+$ 
  with $\proj{h}{\Obsset}=\proj{h'}{\Obsset}$ and
  $U \cap \Reach(\cT,h') \not= \varnothing$.
The task is to show $\Hist_{\cT}(U)=\cL(\cD)$.

Suppose first $h\in \Hist_{\cT}(U)$. Let $x = \delta_{\cD}(\init_{\cD},h)$.
We have to show that $x \in F_{\cD}$.
As $h \in \Hist_{\cT}(U)$ there exist $h'\in (2^P)^+$ 
with $\proj{h}{\Obsset}=\proj{h'}{\Obsset}$
and a state $s \in U \cap \Reach(\cT,h')$. 
By definition of $\Reach(\cT,h')$, there exists an initial, finite path
$\pi$ in $\cT$ such that $\last(\pi)=s$ and
$\trace(\pi)=\proj{h'}{R}$.
By statement (a), this implies $s\in x$. 
By definition of $F_{\cD}$ and because $s \in U$, we obtain $x\in F_{\cD}$.

Vice versa, suppose that $h\in \cL(\cD)$. Thus,
$\delta_{\cD}(\init_{\cD},h) \in F_{\cD}$.
By the definition of $F_{\cD}$, there is a state
$s \in \delta_{\cD}(\init_{\cD},h) \cap U$.
Statement (a) implies the existence of
an initial, finite path $\pi$ in $\cT$ 
and a history $h' \in (2^P)^+$ with  $\last(\pi)=s$,
$\trace(\pi)= \proj{h'}{R}$ and
$\proj{h}{\Obsset} = \proj{h'}{\Obsset}$.
But then $s \in \Reach(\cT,h')$.
This yields $h \in \Hist_{\cT}(U)$.
\end{proof}

%%%%%%%%%%%%%%%%%%%%%%%%%%%%%%%%%%%%%%%%%%%%%%%%%%%%%%%%%%%%%%%%%%%%%%

\begin{lemma}
\label{lemma:properties-T-chi}
 Let $\chi = \StMod{a}{\varphi}$, let
 $\cT_{\chi}$, $\cT_{\chi}'$ be as in the
 step of induction of the model checking procedure and let $R = R_{\chi}$. Then:
 \begin{enumerate}
 \item [(a)] 
    If $h, h' \in (2^{\cP})^*$ with $\proj{h}{P}=\proj{h'}{P}$ then
    $\Reach(\cT_{\chi},h)=\Reach(\cT_{\chi},h') = \Reach(\cT_{\chi}',\proj{h}{P})$.
 \item [(b)]
    If $h \in (2^P)^+$ then
    $\Reach(\cT_{\chi}',h)$ equals the set of all states $((s,O),x_1,\ldots,x_k)$
    in $\cT_{\chi}'$ (or $\cT_{\chi}$)
    such that 
    $s \in \Reach(\cT_a,h)$, $O =\last(h) \cap R$,
    and there exists $h'\in (2^P)^+$
    with $\obs_{\chi}(h)=\obs_{\chi}(h')$ and $x_i=\delta_i(\init_i,h')$
    for all $i\in \{1,\ldots ,k\}$.
 \item [(c)]
    If $z = ((s,O),x_1,\ldots,x_k)$ is a state in $\cT_{\chi}'$ then
     $\Traces(\cT_{\chi}',z) =  \Traces(\cT_a^R,(s,O))$.
 \end{enumerate}
\end{lemma}

\begin{proof}
  Part (a) is obvious. Part (b) can be shown by induction on $|h|$.
  Part (c) follows from the fact that the projection of each path $\pi_{\chi}$
  in $\cT_{\chi}'$ to the $\cT_a^P$-component is a path in $\cT_a^R$
  with the same trace.
  Vice versa, every path $\pi$ in $\cT_a^P$ can be lifted to a path 
  $\pi_{\chi}$ in
  $\cT_{\chi}'$ such that the projection of $\pi_{\chi}$ to the
  $\cT_a^R$-component equals $\pi$. 
\end{proof}

%%%%%%%%%%%%%%%%%%%%%%%%%%%%%%%%%%%%%%%%%%%%%%%%%%%%%%%%%%%%%%%%%%%%%%%%%

\begin{lemma}
 \label{lemma:soundness-default-history-DFA}
  Let $\chi = \StMod{a}{\varphi}$ be a subformula of $\phi$ and
  let $\cD_{\chi}$ 
  be the default history-DFA
  (Definition \ref{def:default-history-DFA}) constructed for $\chi$.
  That is,
  $\cD_{\chi}=\pow(\cT_a^P,\Obsset_{\chi},\Sat_{\cT_a^P}(\exists \varphi))$
  in the basis of induction and
  $\cD_{\chi}=\pow(\cT_{\chi}',\Obsset_{\chi},\Sat_{\cT_{\chi}}(\exists \varphi'))$
  in the step of induction. 
  Then,
  $\cD_{\chi}$ is a history-DFA for $\chi$, i.e.,
  $\Obsset_{\chi}$-deterministic and
  $\cL(\cD_{\chi}) \ = \ 
   \bigl\{ h \in (2^P)^+ : (*,h) \InnerModels{\chi} \chi \bigr\}$.
\end{lemma}

\begin{proof}
By part (b) of Lemma \ref{lemma:properties-default-history-DFA}, 
$\cD_{\chi}$ is $\Obsset_{\chi}$-deterministic.

We now prove by an induction on the nesting depth of standpoint modalities 
that for all standpoint subformulas $\chi$ of $\phi$ it holds that
\begin{center}
  $\cL(\cD_{\chi}) \ = \ 
   \bigl\{ h \in (2^P)^+ : (*,h) \InnerModels{\chi} \chi \bigr\}$
\end{center}
where $\cL(\cD_{\chi})$ is the accepted language of $\cD_{\chi}$.
In what follows, we simply write $\cD$ instead of $\cD_{\chi}$
and denote the components of $\cD$ by
$(\cX_{\cD},\delta_{\cD},\init_{\cD},F_{\cD})$.
Furthermore, we write $R$ instead of $R_{\chi}$.

{\it Basis of induction:}
We consider a subformula
 $\chi = \StMod{a}{\varphi}$ where $\varphi$
is an LTL formula over $P$. 

If $(*,h) \InnerModels{\chi} \chi$ 
then there exists a history $h'\in (2^P)^+$,
a state $s \in \Reach(\cT_a,h')$ 
and a trace $f'\in \Traces^P(\cT_a,s)$ such that
$\first(f')=\last(h')$,
$\obs_{\chi}(h)=\obs_{\chi}(h')$ and $f' \LTLmodels \varphi$.
Let $O=\first(f') \cap R$. Then,  $(s,O)$ is a state in $\cT_a^R$ with
$f'\in \Traces^P(\cT_a^R,(s,O))$ and therefore
$(s,O)\in \Sat_{\cT_a^R}(\exists \varphi)$.
Furthermore, as $O=\last(h') \cap R$ and $s \in \Reach(\cT_a,h')$ 
we have $(s,O)\in \Reach(\cT_a^R,h')$.
Statement (c) of Lemma \ref{lemma:properties-default-history-DFA}
yields $h \in \cL(\cD)$.

Suppose now that $h \in \cL(\cD)$. 
Then, $\delta_{\cD}(\init_{\cD},h) \in F_{\cD}$.
By definition of $F_{\cD}$, 
$\delta_{\cD}(\init_{\cD},h) \cap \Sat_{\cT_a^R}(\exists \varphi)$ is nonempty.
Pick a state 
$(s,O) \in \delta_{\cD}(\init_{\cD},h) \cap \Sat_{\cT_a^R}(\exists \varphi)$.
Then, there exists a word $f' \in \Traces^P(\cT_a^R,(s,O))$
with $f' \LTLmodels \varphi$ and a history $h' \in (2^P)^+$ such that
$(s,O)\in \Reach(\cT_a^R,h')$ and $\obs_{\chi}(h)=\obs_{\chi}(h')$.
But then, $\first(f') \cap R = \last(h') \cap R$ 
(as both agree with $L_a^R(s,O)= O$)
and $s \in \Reach(\cT_a,h')$.

Recall that $R = R_{\chi}=P_a \cup \Obsset_{\chi}$.
If $\first(f')$ and $\last(h')$ are different then we replace the last symbol
of $h'$ with $\first(h')$. 
In this way, we obtain a history $\tilde{h}$ such that
$(f',\tilde{h})$ is a future-history pair with
$\obs_{\chi}(\tilde{h})=\obs_{\chi}(h')=\obs_{\chi}(h)$,
$s\in \Reach(\cT_a,h)=\Reach(\cT_a,\tilde{h})$,
$f'\in \Traces^P(\cT_a^R,(s,O)) \subseteq \Traces^P(\cT_a,s)$ 
and $(f',\tilde{h})\InnerModels{\chi} \varphi$ (as $f' \LTLmodels \varphi$).
This yields $(*,h) \InnerModels{\chi} \chi$.

{\it Step of induction:}
Let $\chi = \StMod{a}{\varphi}$ and let
$\chi_1,\ldots,\chi_k$ be the maximal standpoint subformulas of $\varphi$
and $\cD_i = (X_i,\delta_i,\init_i,F_i)$ the history-DFA for the $\chi_i$'s, which exist by induction hypothesis.
Thus:
\begin{center}
 $\cL(\cD_i)=\{h \in (2^P)^+ : (*,h) \InnerModels{\chi_i} \chi_i\}.$
\end{center}
Moreover, $\InnerModels{\varphi}$ equals
$\InnerModels{\chi_i}$ for $i=1,\ldots,k$.
Thus,
$\cL(\cD_i)=\{h \in (2^P)^+ : (*,h) \InnerModels{\varphi} \chi_i\}$.

Recall that $R_{\chi}= P_a \cup \Obsset_{\chi} \cup \Obsset_1 \cup \ldots \cup \Obsset_k$ where $\Obsset_i = \Obsset_{\chi_i}$, $i=1,\ldots,k$.
In what follows, we will write $R$ instead of $R_{\chi}$.

As $\cD_i$ is a history-DFA for $\chi_i$, $\cD_i$ is $\Obsset_i$-deterministic,
which means that $\proj{h_1}{\Obsset_i}=\proj{h_2}{\Obsset_i}$ implies 
$\delta_{i}(\init_i,h_1)=\delta_i(\init_i,h_2)$.
As $\Obsset_i$ is a subset of $R$, 
the $\cD_i$'s are $R$-deterministic. That is:
\begin{center}
  Whenever $h_1, h_2 \in (2^P)^+$ with $\proj{h_1}{R}=\proj{h_2}{R}$ then
  $\delta_{i}(\init_i,h_1)=\delta_i(\init_i,h_2)$ for $i=1,\ldots,k$.
\end{center}

For every subformula $\psi'$ of 
    $\varphi' = \varphi[\chi_1/p_1,\ldots,\chi_k/p_k]$,
    let $\psi'[p_1/\chi_1,\ldots,p_k/\chi_k]$ 
    be the corresponding subformula
    of $\varphi$. 
  We then have $\varphi'[p_1/\chi_1,\ldots,p_k/\chi_k] = \varphi$
     and $p_i[p_1/\chi_1,\ldots,p_k/\chi_k] = \chi_i$.
Note that $\varphi'$ and its subformulas
  are LTL formulas over $\cP_{\chi}=P\cup \{p_1,\ldots,p_k\}$. 
  In what follows, we simply write $\cP$ rather than $\cP_{\chi}$.
We first show the following statement:

{\it Claim 1:}
    For each subformula $\psi'$ of $\varphi'$, 
    each history $h' \in (2^P)^+$,
    state $z = ((s,O),x_1,\ldots,x_k)$ in $\cT_{\chi}$ 
    such that $x_j = \delta_j(\init_j,h')$ for all $j\in \{1,\ldots,k\}$
    and  $f \in \Traces^{\cP}(\cT_{\chi},z)$ 
    with $\first(f)\cap P =\last(h')$:
    \begin{center}
       $f \LTLmodels \psi'$ \ iff \
       $(\proj{f}{P},h') \InnerModels{\varphi} 
            \psi'[p_1/\chi_1,\ldots,p_k/\chi_k]$ 
    \end{center}
  where $\InnerModels{\varphi}$ and $\InnerModels{\chi}$ 
  agree for the step, public-history and
  pure observation-based semantics, while 
  $\InnerModels{\varphi}$ is $\decrmodels{Q \cap P_a}$ if $\InnerModels{\chi}$ 
  is $\decrmodels{Q}$ 
  and
  $\InnerModels{\varphi}$ is $\incrmodels{Q \cup P_a}$ if $\InnerModels{\chi}$ 
  is $\incrmodels{Q}$.

The proof of Claim 1
is  by structural induction on the syntactic structure
  of $\psi'$. 
  Let us concentrate here on the most interesting case where
  $\psi'=p_i$ for some $i\in \{1,\ldots,k\}$,
  in which case $\psi'[p_1/\chi_1,\ldots,p_k/\chi_k] = \chi_i$.
  Let $h'$, $z = (s,x_1,\ldots,x_k)$ and $f$ be as in Claim 1.
  Then:
  \begin{center}
    $f \LTLmodels p_i$ \ iff \ $p_i\in \first(f)$ \ iff \ $p_i \in L_{\chi}(z)$
    \ iff \ $x_i \in F_i$
  \end{center}
  As we assume here the soundness of the history-DFA $\cD_i$ for $\chi_i$
  (induction hypothsis), 
  we have:  
  \begin{center}
    $x_i \in F_i$
    \ iff \ $h'\in \cL(\cD_i)$ \ iff \ $(*,h') \InnerModels{\chi_i} \chi_i$
  \end{center} 
  Putting things together we obtain:
  \begin{center}
    $f \LTLmodels \psi'=p_i$ \ iff \ $x_i \in F_i$ \ iff \ 
    $(*,h') \InnerModels{\varphi} \chi_i = \psi'[p_1/\psi_1,\ldots,p_k/\psi_k]$
   \end{center}
As $(\proj{f}{P},h') \InnerModels{\varphi} \chi_i$ iff
$(*,h') \InnerModels{\varphi} \chi_i$, 
this yields the claim for the case $\psi'=p_i$
for some $i\in \{1,\ldots,k\}$.
The other case $\psi'=p\in P$ of the basis of induction 
and the step of induction are obvious.

{\it Claim 2:} For each history $h \in (2^P)^+$:
$(*,h)\InnerModels{\chi} \chi$ iff $h \in \cL(\cD_{\chi})$.

By part (a) of Lemma \ref{lemma:properties-default-history-DFA},
$\delta_{\cD}(\init_{\cD},h)$ is the union of the sets
$\Reach(\cT_{\chi}',h')$ where $h'$ ranges over all histories 
$h'\in (2^P)^+$ with $\obs_{\chi}(h)=\obs_{\chi}(h')$ (i.e.,
$\proj{h}{\Obsset_{\chi}}= \proj{h'}{\Obsset_{\chi}}$).

{\it First part of the proof of Claim 2.}
Suppose first that $(*,h)\InnerModels{\chi} \chi$.
Then, there exists a word $h'\in (2^P)^+$, a state $s\in \Reach(\cT_a,h')$
and a trace $f'\in \Traces^P(\cT_a,s)$ such that 
$\obs_{\chi}(h)=\obs_{\chi}(h')$,
$\first(f')=\last(h')$ 
and $(f',h') \InnerModels{\varphi} \varphi$.

Let $O = \first(f') \cap R$. Then, $O=\last(h') \cap R$, 
$(s,O)\in \Reach(\cT_a^R,h')$ and $f'\in \Traces^P(\cT_a^R, (s,O))$.
Furthermore, let
$x_i=\delta_i(\init_i,h')$ for $i=1,\ldots,k$.
Then, $z = ((s,O),x_1,\ldots,x_k)$ is a state in $\cT_{\chi}$ and 
$\cT_{\chi}' = \proj{\cT_{\chi}}{R}$
with $z \in \Reach(\cT_{\chi}',h')$.

As $\obs_{\chi}(h)=\obs_{\chi}(h')$,
part (a) of Lemma \ref{lemma:properties-default-history-DFA} 
yields $z \in \delta_{\cD}(\init_{\cD},h)$.

We pick a path $\pi \in \Paths(\cT_a^R,(s,O))$ 
with $\trace(\pi) = \proj{f'}{R}$.
There is a path $\pi_{\chi}$ in $\cT_{\chi}$ from $z$ where
the projection of $\pi$ to the $\cT_a^R$-components agrees with $\pi$.

Then, $\trace(\pi_{\chi})$ 
is an infinite string over the alphabet
$2^{P_{\chi}}$ (recall that $P_{\chi} = R \cup \{p_1,\ldots,p_k\}$) with 
$\proj{\trace(\pi_{\chi})}{R}=\trace(\pi)=\proj{f'}{R}$.

On the other hand, $f' \in (2^P)^{\omega}$.
As $R = P_{\chi}\cap P$ and $f'$ and $\trace(\pi_{\chi})$ agree elementwise
on the truth values of the propositions in $R$,
we can ``merge'' $f'$ and $\trace(\pi_{\chi})$ to obtain a word 
$f \in \Traces^{\cP}(\cT_{\chi},z)$ such that
$\proj{f}{P}=f'$ and $\proj{f}{P_{\chi}}=\trace(\pi_{\chi})$.

Claim 1 applied to $f$ and $\psi'=\varphi'$ yields:
\begin{center}
  $f \LTLmodels \varphi'$ \ iff \ 
  $(f',h') \InnerModels{\varphi} \varphi$
\end{center}
As $(f',h') \InnerModels{\varphi} \varphi$ (by assumption), we
obtain $f \LTLmodels \varphi'$. As $f \in \Traces(\cT_{\chi},z)$, we get
$z \in \Sat_{\cT_{\chi}}(\exists \varphi')$.
Thus, $z \in \delta_{\cD}(\init_{\cD},h) \cap \Sat_{\cT_{\chi}}(\exists \varphi')$.
This yields
$\delta_{\cD}(\init_{\cD},h) \in F_{\cD}$.
Therefore, $h \in \cL(\cD)$.

{\it Second part of the proof of Claim 2.}
Let us assume now that $h \in \cL(\cD)$.
Statement (c) of Lemma \ref{lemma:properties-default-history-DFA} 
yields that there exists a history $h' \in (2^P)^+$ with
$\obs_{\chi}(h)=\obs_{\chi}(h')$ and 
$\Sat_{\cT_{\chi}}(\exists \varphi') \cap \Reach(\cT_{\chi}',h') \not= \varnothing$.

Pick an element 
$z \in \Sat_{\cT_{\chi}}(\exists \varphi') \cap \Reach(\cT_{\chi}',h')$
and a trace $f \in \Traces^{\cP}(\cT_{\chi},z)$ 
with $f \LTLmodels \varphi'$. (As before, $\cP=P \cup \{p_1,\ldots,p_k\}$.)

Recall that $\cT_{\chi}$ is
a transition system over $P_{\chi} = R \cup \{p_1,\ldots,p_k\}$.
The transition system $\cT_{\chi}'=\proj{\cT_{\chi}}{P}$ 
agrees with $\cT_{\chi}$, the only difference being that the labelings
are restricted to $P$. That is, the labeling function $L_{\chi}'$
of $\cT_{\chi}'$ is given by $L_{\chi}'(z)=L_{\chi}(z)\cap P$.

State $z$ has the form $((s,O),x_1,\ldots,x_k)$ where $(s,O)$ is a
state in $\cT_a^{R}$, 
i.e., $s$ is a state of $\cT_a$ and $O \subseteq R$ with 
$L_a(s) = O \cap P_a$.
As $z \in \Reach(\cT_{\chi}',h')$ we have:
\begin{itemize}
\item
   $x_i = \delta_i(\init_i,h')$ for $i=1,\ldots,k$
\item
   $\last(h') \cap R = L_{\chi}'(z)=O$
\end{itemize}
Using the fact that $f \in \Traces^{\cP}(\cT_{\chi},z)$ we get
$\first(f) \cap P_{\chi} = L_{\chi}(z)$. 
As $R \subseteq P_{\chi}$ we obtain:
\begin{center}
  $\first(f)\cap R \ = \ L_{\chi}(z) \cap R \ = \ L_{\chi}'(z)  \ = \ O 
  \ = \ \last(h') \cap R$
\end{center}
For the case where $\first(f) \cap P \not= \last(h)$, 
we proceed as in the basis
of induction. 
That is, we consider the history $\tilde{h}$ that arises from $h$ by
replacing the last symbol in $h$ with $\first(f) \cap P$.
Then, $\proj{h'}{R}=\proj{\tilde{h}}{R}$. This implies:
\begin{itemize}
\item
  $\obs_{\chi}(h)= \obs_{\chi}(h') = \obs_{\chi}(\tilde{h})$ 
  as $\Obsset_{\chi}\subseteq R$
\item
  $x_i =  \delta_i(\init_i,h')  =  \delta_i(\init_i,\tilde{h})$,
  $i=1,\ldots,k$ as the $\cD_i$'s are $R$-deterministic 
\item
  $z \in \Reach(\cT_{\chi}',h')=\Reach(\cT_{\chi}',\tilde{h})$ 
  and therefore $s \in \Reach(\cT_a,\tilde{h})$ 
\end{itemize}
Moreover, $\last(\tilde{h})=\first(f) \cap P$ (by definition of $\tilde{h}$).
Thus, $(\proj{f}{P},\tilde{h})$ is a future-history pair and:
\begin{center}
  $\proj{f}{P} \in 
  \Traces^P(\cT_{\chi}',z) \ = \ \Traces^P(\cT_a^R,(s,O)) 
  \ \subseteq \ \Traces^P(\cT_a,s)$
\end{center}
By Claim 1 and as $f \LTLmodels \varphi'$, we obtain
$(\proj{f}{P}, \tilde{h}) \InnerModels{\varphi} \varphi$.
 
Putting things together, with $f'=\proj{f}{P}$ we have:
$\obs_{\chi}(h)=\obs_{\chi}(\tilde{h})$, 
$f' \in \Traces^P(\cT_a,s)$, $s \in \Reach(\cT_a,\tilde{h})$ 
and $(f', \tilde{h}) \InnerModels{\varphi} \varphi$.
But this shows that $(*,h) \InnerModels{\chi} \StMod{a}{\varphi}=\chi$.
\end{proof}

%%%%%%%%%%%%%%%%%%%%%%%%%%%%%%%%%%%%%%%%%%%%%%%%%%%%%%%%%%%%%%%%%%%%%

%% file: Sections/app-complexity.tex
%%%%%%%%%%%%%%%%%%%%%%%%%%%%%%%%%%%%%%%%%%%%%%%%%%%%%%%%%%%%%%%%%%%
%%%%%%
%%%%%%     Proofs for the complexity section
%%%%%%
%%%%%%%%%%%%%%%%%%%%%%%%%%%%%%%%%%%%%%%%%%%%%%%%%%%%%%%%%%%%%%%%%%%

\section{Proofs for Section \ref{sec:complexity}}

\label{sec:app-complexity}

%%%%%%%%%%%%%%%%%%%%%%%%%%%%%%%%%%%%%%%%%%%%%%%%%%%%%%%%%%%%%%%%%%%

\subsection{Structure of history-DFAs and time bounds}

We start by providing the proof for 
Lemma \ref{lemma:single-exp-run-time-step-public-decr},
which consists of four parts:
step semantics (Lemma \ref{lemma:step-semantics-history-DFA-single-exp}),
public-history semantics
(Lemma \ref{lemma:public-semantics-history-DFA-single-exp}),
decremental semantics
(Lemma \ref{lemma:decr-semantics-history-DFA-single-exp})
and, finally, the pure observation-based semantics under the
additional assumption that the $P_a$'s are pairwise disjoint
(Lemma \ref{lemma:pobs-disjoint-semantics-history-DFA-single-exp}).

For the step, public-history and decremental semantics
we will make use of Remark \ref{remark:exploit-obs-determinism}.
Recall from Remark \ref{T-a-R} that $R_{\chi}=P_a$ and 
$\cT_a^{R_{\chi}}=\cT_a$
for the step,
pure observation-based and decremental semantics.

If $\varphi$ is an $\StpLTL$ formula, let $\Ag(\varphi)$
denote the set of agents $b\in \Ag$ such that $\varphi$ has a subformula
of the form $\StMod{b}{\psi}$.

We start with the step semantics.
Here, we have $\Obsset_{\chi} = \varnothing$ for all $\chi$.
With the notations used in Remark \ref{remark:exploit-obs-determinism},
the history-DFA $\cD_1,\ldots,\cD_k$ are
$\Obsset_{\chi}$-deterministic. Thus, we can think of
$\cD_{\chi}$ to be defined as a product 
$\pow(\cT_a,\varnothing) \bowtie \cD_1 \bowtie \ldots \bowtie \cD_k$.
But now, each of the $\cD_i$'s also has this shape.
In particular, if \( \chi_i = \StMod{a}{\psi_i} \), then the projection of the first coordinate of the states and the transitions between them in the reachable fragment of \( \cD_i \) matches exactly those of \( \pow(\cT_a, \varnothing) \).
Thus, one can incorporate the information on final states
 and drop $\pow(\cT_a,\varnothing)$ from the product.
The same applies to the history-DFAs for nested subformulas 
$\StMod{a}{\psi'}$ of the $\chi_i$'s.
As a consequence, the default history-DFA $\cD_{\chi}$ 
can be redefined for the step
semantics such that
the state space of $\cD_{\chi}$ is contained in
$\prod_{b\in \Ag(\chi)} 2^{S_b}$.
In particular, 
$\cD_{\chi}$ has the shape $\prod_{a\in \Ag(\chi)} \pow(\cT_a,\varnothing)^{k_a}$
for some $k_a > 0$, and one can merge redundant components if $k_a \geqslant 2$.

\begin{lemma}
\label{lemma:step-semantics-history-DFA-single-exp}
  Under the step semantics, for each standpoint formula
  $\chi$ there is a history-DFA where, besides an additional initial state, the state space is contained in
  $\prod_{a\in \Ag(\chi)} 2^{S_a}$.
\end{lemma}

\begin{proof}
For $a \in \Ag$, $H \subseteq P$ and $s \in S_a$, let 
$\delta_a(s,H)$ denote the set of states $s' \in S_a$ such that
$s \to_a s'$ and $L_a(s')=H \cap P_a$.
For $U \subseteq S_a$, $\delta_a(U,H)=\bigcup_{s \in U} \delta_a(s,H)$.
This definition can be extended inductively to histories by
$\delta_a(U,\varnothing)= U$, 
$\delta_a(U,Hh) = \delta_a( \delta_a(U, H), h)$ for $H \in 2^P$ and
$h \in (2^P)^+$.

Let $\Delta_a : 2^{S_a} \to 2^{S_a}$ be given by
$\Delta_a(U)=\bigcup_{H \subset P} \delta_a(U,H)$. That is, $\Delta_a(U)$
is the set of states $s' \in S_a$ such that $s \to_a s'$
for some state $s \in U$.

For $A \subseteq \Ag$, 
let $D_A = \prod_{a \in A} 2^{S_a}$. 
We define
$\Delta_A : D_A \to D_A$ as follows:
$\Delta_A( (U_a)_{a\in A}) = (\Delta_a(U_a))_{a \in A}$.

Furthermore, let $\delta_A^{\tinystep} : D_A \times 2^P \to D_A$ be given by
$$
  \delta_A^{\tinystep}( (U_a)_{a\in A}, H) \ = \ 
  \Delta_A( (U_a)_{a\in A}) 
$$
for each $H \subseteq P$.
Thus, $\delta_A^{\tinystep}( (U_a)_{a\in A}, H) = 
       \delta_A^{\tinystep}( (U_a)_{a\in A}, H')$
for all $H, H' \subseteq P$.

We now show that for each standpoint formula $\chi = \StMod{a}{\varphi}$
there exists a history-DFA $\cD$ of the form 
\begin{center}
  $\cD = (D_A \cup \{\init_{\cD}\},\delta_A^{\tinystep},\init_{\cD},F_{\chi})$ 
\end{center}
for some $F_{\chi} \subseteq D_A$ where $A = \Ag(\chi)$ and
$\delta_A^{\tinystep}(\init_{\cD},H)= (\{\init_a\})_{a\in A}$ 
for each $H \subseteq P$.  Note that the structure $(D_A \cup \{\init_{\cD}\},\delta^{\tinystep}_A, init_{\cD})$ equals $\prod_{a\in \Ag(\chi)} \pow(\cT_a,\varnothing)$.

Recall from Remark \ref{T-a-R} that $R_{\chi}=P_a$ and that
we may deal with $\cT_a$ rather than $\cT_a^{R_{\chi}}$ in the basis
of induction and in the induction step.

The proof is by induction on the nesting depth of standpoint modalities.
The claim is obvious if $\chi = \StMod{a}{\varphi}$ 
contains no nested standpoint modalities 
(in which case $\varphi$ is an LTL formula)
as then $\Ag(\chi)=\{a\}$ and the default history-DFA
arises by the powerset construction applied to $\cT_a$ and 
$\Obsset = \varnothing$. And indeed, 
$\cD = \pow(\cT_a,\varnothing,U)$
(defined as in Definition \ref{def:default-history-DFA})
is $(D_{\{a\}}\cup \{init_{\{\cD\}}\},\delta_{\{a\}}^{\tinystep}, init_{\{\cD\}}, F)$ for some 
$F \subseteq D_{\{a\}}$.

Suppose now that $\varphi$ has $k$ maximal standpoint subformulas, say
$\chi_1 = \StMod{b_1}{\psi_1},\ldots,\chi_k = \StMod{b_k}{\psi_k}$.
Let $B_i = \Ag(\chi_i)$. 
Then, $\Ag(\varphi) = B_1 \cup \ldots \cup B_k$
and $\Ag(\chi)=\{a\} \cup \Ag(\varphi)$.
By induction hypothesis,
there are history-DFA $\cD_1,\ldots,\cD_k$ 
for $\chi_1,\ldots,\chi_k$ where 
$\cD_i = (D_{B_i}\cup \{\init_{B_i}\},\delta_{B_i}^{\tinystep},\init_{B_i},F_i)$.
We now revisit the definition of the transition system $\cT_{\chi}$.
In Section \ref{sec:model-checking}, we defined the state space $Z_{\chi}$ 
as the set of tuples $(s,x_1,\ldots,x_k)$ 
where $s \in S_a$ and $x_i$ is a state
in $\cD_i$.
Thus, 
the reachable states in $\cT_{\chi}$ have the
form $(s,(U_b)_{b \in B_1},\ldots, (U_b)_{b\in B_k})$ 
and contain redundant components
if there are agents $b\in \Ag$ that belong to multiple $B_i$'s.
That is, we can redefine the state space of $\cT_{\chi}$ to be
$S_a\times \prod_{b \in \Ag(\varphi)} 2^{S_b}$, so $\cT_{\chi}$ can be seen as the product 
$\cT_a \bowtie \cD_{\Ag(\varphi)}$
where $\cD_A$ is the structure $(D_A \cup \{\init_A\},\delta_A^{\tinystep},\init_A)$.
As every state in $\cD_{\Ag(\varphi)}$ has a single successor that is
reached for every input $H \in 2^P$,
the reachable states in the default history-DFA 
$\cD_{\chi}=
 \pow(\cT_{\chi},\varnothing,P_{\chi},\Sat_{\cT_{\chi}}(\exists \varphi'))$ 
defined as in
Definition \ref{def:default-history-DFA} 
have the form $(U,\{ V \})$ where $U \subseteq S_a$ and 
$V \in \prod_{b\in \Ag(\varphi)} 2^{S_b}$.
Thus, $V$ is a tuple $(V_b)_{b\in \Ag(\varphi)}$
where $V_b \subseteq S_b$.
Moreover, if $a \in \Ag(\varphi)$ then $U = V_a$.
Thus, we can redefine the state space of $\cD_{\chi}$ to be
$D_{\Ag(\chi)}$.
In this way, we obtain a history-DFA for $\chi$ 
that has the form 
$(D_{Ag(\chi)} \cup \{init_{Ag(\chi)}\},\delta_{\Ag( \chi)}^{\tinystep},\init_{\Ag(\chi)},F_{\chi})$.
\end{proof}

%%%%%%%%%%%%%%%%%%%%

The situation is similar for the public-history semantics.
Here, $\Obsset_{\chi} = P$ for all $\chi$. Thus, all history-DFA
are $P$-deterministic. With arguments as for the step semantics, 
the history-DFA $\cD_{\chi}$ can be defined in such a way that
the state space of $\cD_{\chi}$ has the shape 
$\prod_{b\in \Ag(\chi)} \pow(\cT_a^P,P)$.
That is, the non-initial states of $\cD_{\chi}$ are elements
of $\prod_{b\in \Ag(\chi)} (2^{S_b}\times 2^P)$. 
Moreover, the reachable fragment
contains only states of the form $(T_b,O_b)_{b\in \Ag(\chi)}$
with $T_b \subseteq S_b$ and $O_b \subseteq P$
where $O_b = O_a$ for all $a,b \in \Ag(\chi)$.
Thus, the state space of $\cD_{\chi}$ can even be reduced to
$(\prod_{b\in \Ag(\chi)} 2^{S_b}) \times 2^P$.
The details are shown in the proof of the following lemma:

\begin{lemma}
\label{lemma:public-semantics-history-DFA-single-exp}
  Under the public-history semantics, for each standpoint formula
  $\chi$ there is a history-DFA where, besides an additional initial state, the state space is contained in
  $(\prod_{a\in \Ag(\chi)} 2^{S_a}) \times 2^P$.
\end{lemma}

\begin{proof}
  The arguments are fairly similar to 
the proof of Lemma \ref{lemma:step-semantics-history-DFA-single-exp}.

For $A \subseteq \Ag$, $a \in \Ag$, let $D_A$ and $\delta_a$ be defined
as in the proof of Lemma \ref{lemma:step-semantics-history-DFA-single-exp}. 
We define a transition function
 $\delta_A : D_A \times 2^P \times 2^P \to D_A \times 2^P$  as follows.
For $H \subseteq P$, $T_a \subseteq S_a$ and $O \subseteq P$, let
\begin{center}
  $\delta_a^{\tinypublic}\bigl( \bigl((T_a)_{a\in A},O\bigr), H \bigr) \ = \ 
      \bigl( (\delta_a(T_a,H))_{a\in A}, H \bigr)$
\end{center}
Furthermore, let 
$I_a^H=\{\init_a\}$ if $L_a(\init_a)= H \cap P_a$ and
$I_a^H=\varnothing$ otherwise.
We now show that for each standpoint formula $\chi = \StMod{a}{\varphi}$
there exists a history-DFA $\cD$ of the form 
\begin{center}
  $\cD = 
  (D_A \times 2^P \cup \{\init_{\cD}\},\delta_A^{\tinypublic},\init_{\cD},F_{\chi})$ 
\end{center}
for some $F_{\chi} \subseteq D_A \times 2^P$ where $A = \Ag(\chi)$ and
$\delta_A^{\tinypublic}(\init_{\cD},H)= \bigl((I_a^H)_{a\in A}, H \bigr)$ 
for each $H \subseteq P$.

  Again, we use an induction on the nesting depth of standpoint modalities. The case when $\chi = \StMod{a}{\varphi}$, where $\varphi$ is an LTL formula is obvious. The default history-DFA arises by the powerset construction applied to $\cT_a^P$ and $\Obsset_{\chi} = P$. Recall the corresponding definition of the transition function of the default history-DFA $\cD$
\begin{center}  
    \begin{tabular}{r}
  	$\delta_{\cD}(x,H)$  =          
  	$\bigr\{ s' \in S : \text{ there exists } s \in x \text{ with }$
  	\ \ \\
  	
  	$s \to s'$ \text{ and } $L(s')  = H  \bigr\}.$
  	\\[1ex]
  \end{tabular}   
\end{center}
  Then the reachable fragment of $\cD$ contains only states which are sets of tuples $(s, O)$ s.t. if $(s_1, O_1), (s_2, O_2) \in x $, where $x$ is a state in the default history-DFA $\cD$,
then $O_1 = O_2$.

  Suppose now that $\varphi$ has $k$ maximal standpoint subformulas, say
  $\chi_1 = \StMod{b_1}{\psi_1},\ldots,\chi_k = \StMod{b_k}{\psi_k}$. 
  Let $B_i=\Ag(\chi_i)$. Then, $\Ag(\varphi) = B_1 \cup \ldots \cup B_k$
  and $\Ag(\chi)=\{a\} \cup \Ag(\varphi)$. 
By induction hypothesis, there are history-DFAs $\cD_i$ for $\chi_i$ such that $\cD_i$ has the shape
$(D_{B_i}\times 2^P \cup\{\init_i\}, \delta_{B_i}^{\tinypublic},\init_i,F_i)$.
Let us revisit the definition of the transition system $\cT_{\chi}$.
  In Section \ref{sec:model-checking}, we defined the state space $Z_{\chi}$ 
  as the set of tuples $((s, O), x_1,\ldots,x_k)$ 
  where $(s, O) \in S_a^P$ and $x_i$ is a state
  in $\cD_i$.
  Thus, 
  the reachable states in $\cT_{\chi}$ have the
  form $((s,O),(T_b, O)_{b \in B_1},\ldots, (T_b, O)_{b\in B_k})$ 
  and contain redundant components
  if there are agents $b\in \Ag$ that belong to multiple $B_i$'s.
  These components as well as the duplicated $O$-components can be merged.
  But then, the reachable states in the default history-DFA 
  $\cD_{\chi}=
  \pow(\cT_{\chi},P,\Sat_{\cT_{\chi}}(\exists \varphi'))$ 
  defined as in
  Definition \ref{def:default-history-DFA} 
  have the form $(U,\{ V \})$ where $U = (U_a, R) \in 2^{S_a}\times 2^P$ and 
  $V = (V_b, R)_{b\in  \Ag(\varphi)}$ with $V_b \subseteq S_b$.
  Moreover, if $a \in \Ag(\varphi)$ then $U_a = V_a$. 
  In this case, the redundant component can be dropped again.
\end{proof} 

%%%%%%%%%%%%%%%%%%%%%%%%%%%%%%%%%%%%%%%%%%%%%%%%%%%%%%%%%%%%%%%%%%%%%%%

We now turn to the decremental semantics. 
If $\chi =\StMod{a}{\varphi}$ then
$\Obsset_{\chi_i} \subseteq \Obsset_{\chi}$ 
for all maximal standpoint
subformulas $\chi_i$ of $\varphi$, and hence their history-DFAs $\cD_i$
are $\Obsset_{\chi}$-deterministic.
In contrast to the step semantics, we cannot guarantee that
  the state space of the generated history-DFA $\cD_{\chi}$ 
  is contained in $\prod_{a\in \Ag(\chi)} 2^{S_a}$. 
  For example, 
  the reachable states of the history-DFA 
  for $\chi=\StMod{a}{ \neXt \StMod{b}{ \neXt \StMod{a}{p} }}$
  under the decremental semantics 
  consist of three components $(T_1,T_2,T_3)$
  where $T_1,T_3 \subseteq S_a$ and $T_2 \subseteq S_b$. 
  Under the decremental semantics, we cannot merge
  the first and the third component as they rely on different observation
  sets ($P_a$ for the first component and $P_a \cap P_b$ 
  for the third component). Thus, $T_1$ and $T_3$ can be different.
However,
 the inductive model checking approach of Section \ref{sec:model-checking}
corresponds to a bottom-up processing of the syntax tree of $v$.
One can now define history-DFAs $\cD_v$ for the nodes in the syntax
tree that have the shape
$\prod_{w} \pow(\cT_{a_w},\Obsset^{\tinydecr}_w)$
where $w$ ranges 
over all nodes in the syntax subtree of the node representing
$\chi$ such that the formula given by $w$ is a standpoint formula
$\StMod{a_w}{\varphi_w}$.

Let, as before, $\phi$ be the $\StpLTL$ formula to be model checked. For each node $v$ in the syntax tree of $\phi$, we write $\phi_v$ for the subformula represented by $v$ and $\Tree(v)$ for the subtree of node $v$. Thus, $\phi  = \phi_r$ for the root node $r$ of the syntax tree. 
All nodes $v$ where $\phi_v$ is a standpoint formula are called standpoint nodes.  $\StpNodes_v$ denotes the set of standpoint nodes in $\Tree(v)$ excluding $v$.
Thus, if $v$ is a standpoint node then 
$\{\phi_w : w \in \StpNodes(v)\}$ 
is the set of all proper standpoint subformulas of $\phi_v$ 
(excluding $\phi_v)$.
The sets $Q^{\tinydecr}_v$ are defined as in Remark \ref{remark:context-dependency-syntax-tree}. Recall that $\Obsset^{\tinydecr}_v=Q^{\tinydecr}_v \cap P_a$ if $v$ is a standpoint node and $\phi_v =\StMod{a}{\varphi_v}$.

Recall that $\pow(\cT,\Obsset)$ denotes the structure defined 
as in Definition \ref{def:default-history-DFA} 
without the declaration of final states.

\begin{lemma}
\label{lemma:decr-semantics-history-DFA-single-exp}
  Under the decremental semantics, for each standpoint node $v$
  there is a history-DFA $\cD_v$ of the shape 
  $\prod_{w \in \StpNodes_v} \pow(\cT_{a_w},\Obsset^{\tinydecr}_w)$.
\end{lemma}

\begin{proof}
  The proof is by induction on $|\StpNodes_v|$.

  The basis of induction is obvious as $|\StpNodes_v|=0$ implies that
  $\phi_v$ has no proper standpoint subformulas.
  That is, $\varphi_v$ is an LTL formula and we can simply
  use the default history-DFA $\cD_v$ defined as in
  Definition \ref{def:default-history-DFA}. 
  
  Assume the induction hypothesis holds for the maximal standpoint subformulas \( w_i \) for \( i \in \{1, \dots, k\} \), where the associated history-DFA are given by
  \[
  \cD_{w_i} = \prod_{w \in \StpNodes_{w_i}} \pow(\cT_{a_w}, \Obsset^{\tinydecr}_w).
  \]
  
  Now, consider the history-DFA \(\cD_\nu\), where $\phi_v =\StMod{a}{\varphi_v}$. Since \( \Obsset^{\tinydecr}_{w_i} \subseteq \Obsset^{\tinydecr}_\nu \) for \( i \in \{1, \dots, k\} \), it follows that all \(\cD_{w_i}\) are \(\Obsset^{\tinydecr}_\nu\)-deterministic. So,  $ \pow(\cT_a \bowtie \cD_{w_{1}} \bowtie \ldots  \bowtie \cD_{w_{k}}, \Obsset^{\tinydecr}_\nu) = \pow(\cT_a , \Obsset^{\tinydecr}_v) \bowtie \cD_{w_{1}} \bowtie \ldots  \bowtie \cD_{w_{k}}$, which completes the proof. 
\end{proof}

	The product 
	\[
	\prod_{w \in \StpNodes_v} \pow(\cT_{a_w}, \Obsset^{\tinydecr}_w)
	\] 
	can be simplified by merging redundant components. Specifically, if \(w, r \in \StpNodes_v\) with 
	\((a_w, \Obsset^{\tinydecr}_w) = (a_r, \Obsset^{\tinydecr}_r)\), then the corresponding terms can be merged.

We now consider the pure observation-based semantics under the additional
assumption that $P_a \cap P_b =\varnothing$ for $a,b\in \Ag$ with $a \not= b$.
Then, for each family $(H_a)_{a\in \Ag}$ where $H_a \subseteq P_a$ there exists
$H \subseteq P$ with $H_a = H \cap P_a$ for all $a \in \Ag$.
Suppose $\chi=\StMod{a}{\varphi}$ is a standpoint subformula
of $\phi$ as in the step of induction and regard the definition of the
transition system $\cT_{\chi}$ which was defined as a product 
$\cT_{\chi} = \cT_a \bowtie \cD_1 \bowtie \ldots \bowtie \cD_k$
where the $\cD_i$'s are history-DFA for the maximal standpoint subformulas
$\chi_i=\StMod{b_i}{\psi_i}$ of $\varphi$. 
(Recall from Remark \ref{T-a-R} that we assume $\cT_a = \cT_a^R$.)
As the $\cD_i$'s are $P_{b_i}$-deterministic, we may think of the $\cD_i$
to be DFA over the alphabet $2^{P_{b_i}}$. 
We can now rephrase the transition relation $\to_{\chi}$ of $\cT_{\chi}$ 
as follows. 
Let $B = \{ b_i : i=1,\ldots,k\}$ and
for $b\in \Ag$, $I_b = \{i \in \{1,\ldots,k\} : b_i = b\}$.
Then, there is a transition  
 $(s,x_1,\ldots,x_k) \to_{\chi} (s',x_1',\ldots,x_k')$ in $\cT_{\chi}$
if and only if $s \to_a s'$ in $\cT_a$ and 
$x_i'=\delta_i(x_i,L_a(s'))$ for $i \in I_a$ and
for each $b \in B \setminus \{a\}$ there exists $\tilde{H}_b \subseteq P_b$
such that 
$x_i'= \delta_i(x_i,\tilde{H}_b)$ for $i \in I_b$.
Thus, the local transitions of $\cD_i$ for $i \notin I_a$ do not depend on
the local state $s$ of $\cT_a$. Furthermore, there is only synchronisation
between DFAs $\cD_i, \cD_j$ where $b_i=b_j$.   
This observation can be used to show by induction on the nesting depth
of standpoint formulas $\chi=\StMod{a}{\varphi}$ 
that the definition of default history-DFAs $\cD_{\chi}$ can be modified
such that the state space is contained in 
$2^{S_a} \times \prod_{b \in \Ag(\varphi)} 2^{S_b}$.

\begin{lemma}
 \label{lemma:pobs-disjoint-semantics-history-DFA-single-exp}
  Suppose that the $P_a$'s are pairwise disjoint.
  Under the pure observation-based semantics, for each standpoint formula
  $\chi = \StMod{a}{\varphi}$ there is a
  history-DFA $\cD_{\chi}$ where the state space is contained in
  $2^{S_a} \times \prod_{b \in \Ag(\varphi)} 2^{S_b}$ along with an additional initial state.
 
\end{lemma}

\begin{proof}
For $a \in \Ag$ and $B \subseteq \Ag$, 
let $\delta_a : 2^{S_a} \times (2^P)^* \to 2^{S_a}$,
$\Delta_a : 2^{S_a} \to 2^{S_a}$ and $\Delta_B : D_B \to D_B$ 
(where $D_B = \prod_{b\in B} 2^{S_b}$) and
$\delta_B^{\tinystep} : D_B \times 2^P \to D_B$
be defined as in
the proof of Lemma \ref{lemma:step-semantics-history-DFA-single-exp}.

Given a pair $(a,B)\in \Ag \times 2^{\Ag}$, we define
$D'_{a,B}= 2^{S_a} \times D_B$ and
$\delta_{a,B}^{\tinypobs} : D'_{a,B} \times 2^P \to D'_{a,B}$ as follows:

$$
  \delta_{a,B}^{\tinypobs}\bigl( (U, (U_b)_{b \in B}), H \bigr) 
  \ = \ 
  \bigl( \delta_a(U,H), \Delta_B^{\tinystep}((U_b)_{b\in B}) \bigr)
$$
for all $H \subseteq P$, $U \subseteq S_a$ and $U_b \subseteq S_b$ for $b\in B$.
Furthermore, we define $$
\delta_{a,B}^{\tinypobs}( init_{a, B}, H)=(\{\init_a\}, (\{\init_b\})_{b\in B}),
$$ when $L_a(init_a) = H\cap P_a$.

We shall write $\cD_{a,B}$ for the structure
$(D_{a,B},\delta_{(a,B)}^{\tinypobs},\init_{a,B})$, where $D_{a,B} = D'_{a,B}\cup \{init_{a, B}\}$
and $\cD_B$ for  $(D_{B}\cup \{\init_{B}\},\delta_{B}^{\tinystep},\init_{B})$,
where $\delta_{B}^{\tinystep}(\init_B, H) = (\{\init_b\})_{b\in B}$, for each $H\subseteq P$.

We now show by induction on the nesting depth of standpoint formulas
that each standpoint formula
$\chi = \StMod{a}{\varphi}$ has a history-DFA of the form 
$(D_{a,E},\delta_{a,E}^{\tinypobs},\init_{a,E},F_{\chi})$
where 
$E = \Ag(\varphi)\setminus \{a\}$ 
if $\varphi$ has no subformula $\StMod{a}{\psi}$ 
where the alternation depth is strictly smaller than $\ad(\varphi)=\ad(\chi)$
and $E = \Ag(\varphi)$ otherwise.
That is, $a \in E$ if and only if 
if $\varphi$ has a subformula $\StMod{a}{\psi}$ 
with $\ad(\StMod{a}{\psi}) < \ad(\chi)$

The remaining steps are now similar to the proof of
Lemma \ref{lemma:step-semantics-history-DFA-single-exp}.
The basis of induction is obvious as we can deal with 
the history-DFA $\pow(\cT_a,P_a,\ldots)$ 
defined as in Definition \ref{def:default-history-DFA}.

Suppose now that $\varphi$ has $k$ maximal standpoint subformulas, say
$\chi_1 = \StMod{b_1}{\psi_1},\ldots,\chi_k = \StMod{b_k}{\psi_k}$.
By induction hypothesis,
there are history-DFA $\cD_1,\ldots,\cD_k$ 
for $\chi_1,\ldots,\chi_k$ that have the form
$\cD_i = (D_{b_i,C_i},\delta_{b_i,C_i}^{\tinypobs},\init_{b_i,C_i},F_i)$
where $C_i = \Ag(\psi_i)$ or $C_i = \Ag(\psi_i)\setminus \{b_i\}$,
depending on whether $\psi_i$ contains a standpoint subformula for agent $b_i$
with alternation depth strictly smaller than $\ad(\chi_i)$.

We now revisit the definition of the transition system $\cT_{\chi}$.
In Section \ref{sec:model-checking}, we defined the state space $Z_{\chi}$ 
as the set of tuples $(s,x_1,\ldots,x_k)$ 
where $s \in S_a$ and $x_i$ is a state
in $\cD_i$. (Recall from Remark \ref{T-a-R} that we may assume that
$\cT_a^R = \cT_a$.)
Thus, 
the reachable states in $\cT_{\chi}$ have the
form $(s,(b_1,(U_c)_{c \in C_1}),\ldots, (b_k,(U_{c\in C_k})))$ 
and contain redundant components
if there are agents $c\in \Ag$ that belong to multiple $C_i$'s
or if an agent $b$ coincides with multiple $b_i$'s.
Thus, we may can redefine the state space of $\cT_{\chi}$ to be
$S_a\times \prod_{b \in B} 2^{S_b} \times
     \prod_{c \in C} 2^{S_c}$ where $B = \{b_1,\ldots,b_k\}$ and
$C = C_1 \ldots \cup C_k$. Note that 
$B \cap C \not= \varnothing$ as well as $a \in B \cup C$ is possible.

Using the assumption that $P_a \cap P_b = \varnothing$ for $a \not= b$,
the transitions in $\cT_{\chi}$ are given by:
$$
  \bigl( s, (T_b)_{b\in B}, (U_c)_{c\in C} \bigr) \to_{\chi}
  \bigl( s', (T_b')_{b\in B}, (U_c')_{c\in C} \bigr)
$$
if and only if the following conditions hold:
\begin{itemize}
\item 
   $s \to_a s'$
\item 
   $(U_c')_{c\in C} = \Delta_C ((U_c)_{c\in C}$
\item
    for each $b \in B$:
    \begin{itemize}
    \item if $b = a$ then $T_a'=\delta_a(T_a,L_a(s'))$ 
    \item if $b \not= a$ then $T_b' = \Delta_b(T_b)$ 
           (as $P_a \cap P_b =\varnothing$)
    \end{itemize}
\end{itemize}
Note that $\delta_a(T_a,L_a(s')) = \delta_a(T_a,H)$
          for each $H \subseteq P$ with $H \cap P_a = L_a(s')$ by definition
of $\delta_a$.

Thus, if $b \in B \cap C$ and $b \not= a$ 
then $U_b = T_b$ for all reachable states 
$\bigl( s, (T_b)_{b\in B}, (U_c)_{c\in C} \bigr)$
in $\cT_{\chi}$. Hence, we can further merge components of
the states in $\cT_{\chi}$ and redefine $\cT_{\chi}$ 
such that its state space $Z_{\chi}$ is contained
\begin{itemize}
\item 
    in $S_a\times  \prod_{c \in E} 2^{S_c}$ if $a \notin B$
\item 
    in $S_a\times 2^{S_a} \times \prod_{c \in E} 2^{S_c}$ if $a \in B$
\end{itemize}
where $E = (B \setminus \{a\}) \cup C$.
Thus, $\cT_{\chi}$ can be redefined as the product 
$\cT_a \bowtie \cD_{a,E}$ if $a \in B$ and
as the product $\cT_a \bowtie \cD_{E}$ if $a \notin B$.

We now consider the powerset construction used to define the
default history-DFA $\cD_{\chi}= \pow(\cT_{\chi}',P_a,\ldots)$.
(Recall that $\cT_{\chi}'$ is essentially the same $\cT_{\chi}$ except
the labelings are restricted to $P$.)
Both $\cD_{a,E}$ and $\cD_E$ ar $P_a$-deterministic.
This yields:
\begin{itemize}
\item
  If $a \in B$ then the reachable states of 
  $\cD_{\chi} = \pow(\cT_a \bowtie \cD_{a,E}, P_a,\ldots)$ have the form
  $(T,\{ (U,(U_c)_{c\in E}\})$ where  $U=T$.
\item
  If $a \notin B$ then the states of 
  $\cD_{\chi} = \pow(\cT_a \bowtie \cD_{E}, P_a,\ldots)$ have the form
  $(T,\{ (U_c)_{c\in E}\})$.
\end{itemize}
In both cases, we can redefine the default history-DFA such that
$\cD_{\chi}$ has the form 
$(D_{a,E},\delta_{a,E}^{\tinypobs},\init_{a,E},F_{\chi})$.
\end{proof}

%%%%%%%%%%%%%%%%%%%%%%%%%%%%%%%%%%%%%%%%%%%%%%%%%%%%%%%%%%%%%%%%%

\subsection{Embedding of $\StpLTL$ into LTLK}

  In Lemma \ref{lemma:embdding-pobs-in-LTLK}
  we saw that standpoint LTL under 
  the pure observation-based semantics can be embedded into
  LTLK. 
  We now show  (see Lemma \ref{lemma:step-public-in-LTLKone}
  and Lemma \ref{lemma:embedding-incr-in-LTLK} below) that
  the embedding proposed in Lemma \ref{lemma:embdding-pobs-in-LTLK}
  can be modified for the other four semantics.
  While the proposed embedding for the pure observation-based
  semantics is fairly natural and preserves
  the set of agents, the embeddings presented in the proofs of
  Lemma \ref{lemma:step-public-in-LTLKone}
  and Lemma \ref{lemma:embedding-incr-in-LTLK} modify the agent sets.
  In the cases of the step and public-history semantics
  it suffices to deal with a single agent in the constructed LTLK structure.
  In the case of the incremental and decremental semantics,
  we add new agents which can be seen as copies of the agents $a\in \Ag$
  with different indistinguishable relations in the LTLK structure.

\begin{lemma}
 \label{lemma:step-public-in-LTLKone}
   For $*\in \{\step,\public\}$,
   the $\StpLTL^*$ model checking problem 
   is polynomially reducible to the LTLK$_1$ model checking
   problem.
\end{lemma}

\begin{proof}
  We revisit the translation of
  $\StpLTL$ structures $\fT$ and $\StpLTL$ formulas $\phi$
  into an LTLK structure $\ltlk{\fT}$ and an LTLK formula $\ltlk{\phi}$
  presented in the proof of Lemma \ref{lemma:embdding-pobs-in-LTLK}
  for the pure observation-based semantics.

  For the step semantics, we redefine the LTLK structure 
  $\ltlk{\fT}$ as follows. The transition system $\cT'$
  is the same as in the proof of Lemma \ref{lemma:embdding-pobs-in-LTLK},
  but now we deal with a single agent, say $\alpha$, and define
  $\sim_{\alpha}$ to be the equivalence relation that identifies all states
  in $\cT'$.
  The translation of the $\StpLTL$ formula $\phi$ into an LTLK formula
  $\ltlkstep{\phi}$ is the same as for the pure observation-based setting,
  except that we replace $\psi= \StMod{a}{\varphi}$ with
  $\primed{\psi} = 
    \overline{K}_{\alpha}{(\primed{\varphi} \wedge \Box \neg \bot_a)}$.

  In this way, we obtain an LTLK formula $\ltlkstep{\phi}$ over a
  singleton agent set. 
  Such formulas have alternation depth at most 1.
  Thus, $\ltlkstep{\phi}$ is an LTLK$_1$ formula.
  Soundness in the sense that $\fT \stepmodels \phi$ if and only if
  $\ltlk{\fT} \LTLKmodels \ltlk{\phi}$ can be shown as in the proof
  of Lemma \ref{lemma:embdding-pobs-in-LTLK}
  (see Lemma \ref{app:lemma:embdding-pobs-in-LTLK}).

  For the public-history semantics,
  the reduction is almost the same as for the step semantics,
  except that we deal  with the equivalence relation
  $\sim_{\alpha}$ that identifies those states
  $\sigma$ and $\theta$ in $\cT'$ where 
  $L'(\sigma)=L'(\theta)$.
\end{proof}

The statement of Theorem \ref{thm:step-public-PSPACE-completeness}
is now a consequence of Lemma \ref{lemma:step-public-in-LTLKone}
in combination with the known PSPACE-completeness result for
LTLK$_1$.

%%%%%%%%%%%%%%%%%%%%%%%%%%%%%%%%%%%%%%%%%%%%%%%%%%%%%%%%%%%%%%%%%%%%%%%%%

\begin{lemma}
 \label{lemma:embedding-incr-in-LTLK}
  Let  $N=|\Ag|$, $d \geqslant 1$ and $*\in \{\decr,\incr\}$.
  The $\SLTL{*}{d}$ model checking problem 
  is polynomially reducible to the LTLK$_M$ model checking
  problem where $M = \min \{N,d\}$.
\end{lemma}

\begin{proof}
Let $\fT$ be a given $\StpLTL$ structure and $\phi$ a 
$\StpLTL_d$ formula.
We modify the construction presented in the proof of
Lemma \ref{lemma:embdding-pobs-in-LTLK} 
for the pure observation-based
semantics and redefine the LTLK structure $\ltlk{\fT}$ 
and the LTLK formula $\ltlk{\phi}$.

The transition system $\cT'$ of $\ltlk{\fT}$ is the same as in the
proof of Lemma \ref{lemma:embdding-pobs-in-LTLK}.
But now we deal with more agents. The idea is to add agents for the standpoint
subformulas $\chi = \StMod{a}{\varphi}$ of $\phi$ where $Q^{\tinyincr}_{\chi}$
is a proper superset of $P_a$.

We present the details for the incremental semantics.
The arguments for the decremental semantics are analogous and omitted here.

We consider the syntax tree of $\phi$. For each node $v$ in the syntax tree, 
let $\phi_v$ be the formula represented by $v$.
Let $V$ be the set of nodes $v$
where $\phi_v$ is a standpoint formula, say
$\phi_v = \StMod{a_v}{\phi_w}$ for the child $w$ of $v$. 
Given $v \in V$, let $A_v$ denote the set
of all agents $a \in \Ag$ such that the unique path from the root to  
$v$ contains at least one node $w \in V$ with $a_w=a$.
(With the notations used in the paragraph before
Lemma \ref{lemma:N-EXP-incremental-semantics} we have
$A_v=\Ag_v \cup \{a\}$.)

The observation set for the occurrence of $\phi_v$ as a subformula of $\phi$
represented by node $v$ is 
$\Obsset_{v}^{\tinyincr} = P_{A_v}$
where $P_A = \bigcup_{a \in A} P_a$  for $A \subseteq \Ag$. 
The set $Q_v^{\tinyincr}$ defined in 
Remark \ref{remark:context-dependency-syntax-tree} is
either $P_{A_v}$ or $P_{A_v}\setminus \{a_v\}$.
More precisely,
$Q_v^{\tinyincr} = \varnothing$ if $v$ is the root. For all other nodes $v$,
if $w$ is the father of $v$, then $Q_v^{\tinyincr} = P_{A_v} \setminus \{a_v\}$ 
if $a_v \notin P_{A_w}$ and $Q_v^{\tinyincr} = P_{A_v}$ if $a_v \in P_{A_w}$.

The new agent set of the constructed LTLK structure $\ltlk{\fT}$ is given by:
\begin{center}
  $\Ag' \ = \ 
  \bigl\{ A_v : v \in V, A_v \not= \varnothing \bigr\}$
\end{center}
Thus, $\Ag'$ is a subset of $2^{\Ag}$. 
Its size $|\Ag'|$ is bounded by the size of the syntax tree of $\phi$
and therefore linearly bounded in the length of $\phi$.

The equivalence relations $\sim_{A}$ for agent $A\in \Ag'$
is given by $\sigma \sim_{A} \theta$ if and only if 
$L'(\sigma) \cap P_A^{\bot} = L' (\theta) \cap  P_A^{\bot}$
where $P_A^{\bot}= \bigcup_{a\in A} P_A^{\bot} = P_A \cup \{\bot_a : a \in \Ag \}$.

Intuitively, agent $a$ in $\fT$
corresponds to agent $\{a\}$ in $\ltlk{\fT}$ while the new agents $A$ with
$|A|\geqslant 2$ in $\ltlk{\fT}$ can be seen as coalitions of 
agents who share their
information about the history.

The translation $\varphi \leadsto \primed{\varphi}$ of $\StpLTL$ formulas
$\varphi$ into ``equivalent'' LTLK formulas $\primed{\varphi}$ 
presented in the proof of Lemma \ref{lemma:embdding-pobs-in-LTLK}
is now replaced
with a translation of the $\StpLTL$ formulas $\phi_v$ for the nodes
in the syntax tree of $\phi$ into LTLK formulas.
If $v$ is a leaf (i.e., $\phi_v  \in \{\true\} \cup P$) then $\phi_v'=\phi_v$.
Let $v$ now be an inner node.
If 
$\phi_v = \phi_w \wedge \phi_u$ or $\phi_v = \phi_w \Until \phi_u$
(where $w$ and $u$ are the childs of $v$) then
$\phi_v' = \phi_w' \wedge \phi_u'$ or $\phi_v' = \phi_w' \Until \phi_u'$,
respectively.
The definition is analogous for the cases 
$\phi_v = \neg \phi_w$ or $\phi_v = \neXt \phi_w$.
Suppose now that $v \in V$, i.e., $\phi_v = \StMod{a_v}{\phi_w}$
for the child $w$ for $v$.
Then, we define:
\begin{center}
   $\phi_v' \  = \ 
    \overline{K}_{A_v}{(\phi_w' \wedge \Box \neg \bot_{a_v})}$
\end{center}
For each path $\pi$ in $\cT'$, each $n \in \Nat$ and node $v$ 
in the syntax tree
of $\phi$ we have (where we write $Q_v$ instead of $Q_v^{\tinyincr}$):
\begin{center}
  $(\pi,n) \LTLKmodels \phi_v'$
  \ iff \ 
  $(\proj{\trace(\suffix{\pi}{n})}{P},
    \proj{\trace(\prefix{\pi}{n})}{P}) \incrmodels{Q_v} \phi_v$
\end{center}
Finally, we define
$\ltlk{\phi}$ as in the proof of
Lemma \ref{lemma:embdding-pobs-in-LTLK} by
$\ltlk{\phi}= \Box \neg \bot_0 \to \phi'$
where $\phi' = \phi_v'$ for the root node $v$ for the syntax tree of $\phi$.
It is now easy to see that
$\fT \incrmodels{} \phi$ iff
$\ltlk{\fT} \LTLKmodels \ltlk{\phi}$.
Furthermore,
the alternation depth of 
$\ltlk{\phi}$ is bounded by both $N=|\Ag|$ and the alternation depth of $\phi$.
\end{proof}

%%%%%%%%%%

Corollary \ref{cor:incr-N-1-EXPSPACE} is now a consequence of
Lemma \ref{lemma:embedding-incr-in-LTLK} 
in combination with the results of
\cite{MC-LTLK-and-beyond-2024}.

%%%%%%%%%%%%%%%%%%%%%%%%%%%%%%%%%%%%%%%%%%%%%%%%%%%%%%%%%%%%%%%%%

\subsection{Generic model checking algorithm with improved space complexity}

\label{sec:m-minus-1-EXPSPACE}

The complexity of 
the generic $\StpLTL^*$ model checking algorithm of
Section \ref{sec:model-checking}
is $m$-fold exponentially time-bounded
with $m=1$ for $* \in \{\step,\public,\decr\}$ (see Lemma \ref{lemma:single-exp-run-time-step-public-decr}),  $m = |\Ag|$ for $* = \incr$ (see Lemma \ref{lemma:N-EXP-incremental-semantics})
and $m = \ad(\phi)$ for $*=\pobs$ (see Lemma \ref{d-exp-time-bound}).
To match the space bounds stated in
Corollary 5.5 for $\pobsmodels$ and $\incrmodels{}$, Theorem 5.6 and
Corollary \ref{cor:incr-N-1-EXPSPACE},
the algorithm of Section \ref{sec:model-checking}
can be adapted to obtain an $(m{-}1)$-EXPSPACE algorithm.

Let us briefly sketch the ideas for the pure observation-based semantics.
The idea for $\StpLTL_1$ formulas is to combine classical on-the-fly
automata-based LTL model checking techniques with an on-the-fly
construction of history automata, 
similar to the
techniques proposed in \cite{MC-LTLK-and-beyond-2024} for CTL*K.
This yields a polynomially-space bounded algorithm for $\StpLTL$ formulas
of alternation depth 1.

Suppose now that $d=\ad(\phi)\geqslant 2$.
Let $\chi_1,\ldots\chi_{\ell}$ be the maximal standpoint subformulas
of $\phi$ with $\ad(\chi_i) = d$, $i=1,\ldots,\ell$,
and $\chi_{\ell+1},\ldots,\chi_k$ the maximal standpoint subformulas
of $\phi$ with  $\ad(\chi_i) < d$, $i=\ell{+}1,\ldots,k$.
We apply the model checking
algorithm of Section \ref{sec:model-checking} to generate history-DFAs
$\cD_{\psi}$ for all standpoint subformulas $\psi$ of $\phi$ where either
$\psi \in \{\chi_{\ell {+}1},\ldots,\chi_k\}$ 
or $\psi$ is a standpoint subformula of one of the formulas
$\chi_i$, $i \in \{1,\ldots,\ell\}$ with $\ad(\psi) = d{-}1$ and maximal with
this property (i.e., there is no standpoint 
subformula $\psi'$ of $\chi_i$ where
$\psi$ is a proper subformula of $\psi'$ and $\ad(\psi') = d{-}1$).
Let $\Psi$ denote the set of these subformulas $\psi$ of $\phi$.
The sizes of the history-DFAs $\cD_{\psi}$ 
for $\psi \in \Psi$ are $(d{-}1)$-fold
exponentially bounded.
We now proceed in a similar way as we did in the induction step 
of the model checking algorithm
and introduce pairwise distinct, fresh atomic propositions $p_{\psi}$ 
for each $\psi \in \Psi$.
Then, $\phi' = \phi[\psi/p_{\psi} : \psi \in \Psi]$ is a $\StpLTL$ formula
over $\cP = P \cup \{p_{\psi} : \psi \in \Psi \}$
and has alternation depth 1. We then modify the
$\StpLTL$ structure $\fT$ by replacing
$\cT_0$ with the product of $\cT_0$ with the history-DFAs
for $\chi_1,\ldots,\chi_{\ell}$ and modifying the $\cT_a$'s accordingly.
The size of the resulting structure $\fT'$ 
is $(d{-}1)$-fold exponentially bounded.
We finally run a polynomially space-bounded algorithm to $\fT'$ and $\phi'$.